\documentclass[10pt,a4paper]{article}
\usepackage[utf8]{inputenc}
\usepackage[T1]{fontenc}
\usepackage{microtype}
\usepackage{fourier}
\usepackage{bm}
\usepackage[margin=1in]{geometry}
\usepackage{tikz}
\usepackage[margin=10pt,font=small,labelfont=bf,labelsep=period]{caption}
\usepackage{amsmath}
\usepackage{amsthm}
\usepackage{amssymb}
\usepackage{mathscinet}
\usepackage{pdfpages} 
\usepackage{empheq}
\usepackage{float}
\usepackage{subcaption}

\numberwithin{equation}{section}

\usepackage[nottoc,notlot,notlof]{tocbibind}

\usepackage{aliascnt}
\usepackage[colorlinks,linkcolor=blue,citecolor=blue]{hyperref}
\usepackage[nameinlink,capitalise]{cleveref}

\theoremstyle{plain}

\newtheorem{theorem}{Theorem}[section]

\newaliascnt{lemma}{theorem}
\newtheorem{lemma}[lemma]{Lemma}
\aliascntresetthe{lemma}

\newaliascnt{corollary}{theorem}
\newtheorem{corollary}[corollary]{Corollary}
\aliascntresetthe{corollary}

\theoremstyle{definition}

\newaliascnt{definition}{theorem}
\newtheorem{definition}[definition]{Definition}
\aliascntresetthe{definition}

\newaliascnt{example}{theorem}
\newtheorem{example}[example]{Example}
\aliascntresetthe{example}

\newaliascnt{remark}{theorem}
\newtheorem{remark}[remark]{Remark}
\aliascntresetthe{remark}

\newaliascnt{assumption}{theorem}

\aliascntresetthe{assumption}

\newaliascnt{proposition}{theorem}
\newtheorem{proposition}[proposition]{Proposition}
\aliascntresetthe{proposition}

\newcommand{\R}{\mathbf{R}}

\renewcommand{\epsilon}{\varepsilon}
 
\newcommand\Cl{\mathcal{C}\ell}

\newcommand{\abs}[1]{\left\lvert #1 \right\rvert}

\newcommand{\bo}{{\bf 1}}
\renewcommand{\l}{\lambda}

\newcommand{\Ithree}{I_3}
\newcommand{\asqrt}{\sqrt{\Ithree}}

\DeclareMathOperator{\sgn}{sgn}
\DeclareMathOperator{\tr}{tr}
\DeclareMathOperator{\arcsinh}{arcsinh}
\DeclareMathOperator{\sn}{sn}

\title{ Two-peakon dynamics in the Clifford algebra generalization of the Camassa–Holm equation }
\author{Alexander Karlson\thanks{Department of Mathematics \& Statistics and Centre for Quantum Topology and Its Applications (quanTA), University of Saskatchewan, Saskatoon, SK, S7N 5E6, CANADA;
E-mail: ytw374@mail.usask.ca}\and Jacek Szmigielski \thanks{Department of Mathematics \& Statistics and Centre for Quantum Topology and Its Applications (quanTA), University of Saskatchewan, Saskatoon, SK, S7N 5E6, CANADA; 
  E-mail: szmigiel@math.usask.ca} }

\date{\today}
\makeatletter
\hypersetup{%
  pdfauthor={Jacek Szmigielski},
  pdftitle={\@title},
}
\makeatother

\begin{document}

\maketitle

\begin{abstract}

We study the dynamics of a two-component perturbation of the
Camassa--Holm equation arising from a reformulation of the
Euler--Bernoulli beam problem, recently extended to a general
Clifford algebra setting. We focus on the original case associated
with a Clifford algebra with two generators and Minkowski signature,
for which the resulting equation admits nonsmooth soliton solutions
(peakons) carrying internal degrees of freedom.

We investigate analytically and numerically the dynamics of a
two-peakon solution and establish the existence of a synchronized
exchange of energy between spatially separated peaks, a phenomenon
previously observed only numerically. We obtain a complete
description of the long-time dynamics: the amplitudes approach a
periodic orbit determined by the spectral invariants, and the
resulting hidden periodicity governs the persistent exchange between
the two peakons. The limiting orbit and its averaged dynamics are
described explicitly in terms of elliptic functions and complete
elliptic integrals.

We also derive an exact identity relating the peak separation to the
accumulated imbalance of the two amplitudes. This identity, combined
with a frozen-parameter comparison argument, yields exponential
decay of the interaction between the peakons and exponential
convergence of the amplitude variables to the limiting periodic
orbit. Moreover, once the asymptotic regime is reached, the
separation increases from one period of the internal oscillation to
the next, even though its instantaneous rate may continue to change
sign.

These results reveal a dynamical feature absent from the scalar
Camassa--Holm equation: a persistent oscillatory transfer of energy
between increasingly separated peakons, coupled with quantitatively
controlled asymptotic decoupling.

\end{abstract}
\tableofcontents{}
\section{Introduction} 
In this paper, we examine a system of PDEs introduced in \cite{BealsSzm1} \footnote{An equivalent system appeared earlier in the work of Geng and Wang \cite{geng-wang:Coupled-CH}[equation 1.2 therein].   }
\begin{align} 
&m_t=(um)_x+u_x m-vm, \qquad &n_t&=(un)_x+u_x n+vn, \label{eq:2CH}\\
&u_{xx}-4u=n+m,  \qquad &v_x&=n-m. \label{eq:2CHconstraints}
\end{align} 

This system was derived in the context of isospectral deformations of the Euler--Bernoulli beam problem.  Yet, interestingly, this classical problem can be viewed as equipping the 
Camassa--Holm (CH) equation \cite{camassa-holm, CHH1994}, written below in the original notation, 
\begin{equation} \label{eq:CH} 
m_t+(um)_x+u_xm=0, \qquad m= u-u_{xx} 
\end{equation}
with an internal structure conceptually similar to the transition from the Klein--Gordon scalar equation to the spinor Dirac equation.  
We note that if $v=0$ then equations \eqref{eq:2CH} and \eqref{eq:2CHconstraints} decouple into two independent (rescaled) CH equations.  We will subsequently refer to $v$ as a \emph{coupling field}.  

One of the most interesting features of the original CH equation is the existence of non-smooth solitons, known as peakons. Using the notation (differing from ours only by scaling) from the original paper 
\cite{camassa-holm}, the peakons are obtained from the Ansatz: 
\begin{equation} \label{eq:uwave} 
u(x,t)=\sum_{j=1}^N m_j(t) e^{-\abs{x-x_j(t)}},  
\end{equation} 
which, after substitution into \eqref{eq:CH} and interpreting equations in a weak sense, yields a Hamiltonian system of ODEs: 
\begin{equation} \label{eq:CHpeakonsH}
\dot x_j=\frac{\partial H}{\partial m_j}, \qquad \dot m_j=-\frac{\partial H}{\partial x_j}, 
\end{equation} 
where $H=\frac12 \sum_{j=1, k=1}^N m_j m_k e^{-\abs{x_j-x_k}}$.  

There is a wealth of literature on peakons, starting with \cite{camassa-holm}, but perhaps the most sensible approach is to refer to the review paper written by one of us \cite{lundmark-szmigielski:review}, where the reader can find 
an extensive list of references (336 references to be exact)  to work on peakons over the last thirty years.  

There have been multiple generalizations of the CH equation in the past. The literature on multicomponent generalizations is also so vast that it cannot be done justice to in a few paragraphs. Nevertheless, we select a few representative works to illustrate the scope of generalizations. The first modification, to the best of our knowledge, was in the paper by Olver and Rosenau \cite{Olver-Rosenau:Tri}, followed by Chen, Liu and Zhang \cite{CLZ} and by Falqui \cite{Falqui_2006}.  Some of the bi-Hamiltonian context for the paper \cite{CLZ} is also contained in the paper by Liu and Zhang \cite{LiuZhang2005}.   Eventually, this generalization became known as the 2CH model. 
We should also mention some early work of Popowicz, for example \cite{Popowicz2007}, and its supersymmetric counterpart \cite{Popowicz-super}.  These early papers used 
as methodology, some combination of the Lax pair modification and/or a deformation of the CH bi-Hamiltonian structure.  
The physical justification for the 2CH model was presented in the paper by Constantin and Ivanov \cite{ConstantinIvanov2008}; these authors also studied the short-wave limit 
of the 2CH, observing that limit admits peakon solutions, in contrast to the full 2CH system.  A preliminary classification of two-component CH systems was given by Hone, Novikov, and Wang in \cite{HNW} and was later extended to include the system we study in this paper in \cite{HNS}.  
Further generalizations followed, for example, Holm and Ivanov \cite{HolmIvanov2010}, and a series of papers 
by Qiao and his collaborators \cite{XiaQiao2009,XZQ2015,XQ2016}.  It is worth pointing out that the system equivalent to our \eqref{eq:2CH} was presented as one of three interesting 
peakon equations associated with the $4\times 4$ Lax pair by Geng and Wang in \cite{geng-wang:Coupled-CH}.  Finally, it is fair to mention that many of the structures arising in the Clifford-algebra-based generalization of the Euler--Bernoulli beam problem proposed in \cite{BealsSzm1, BealsSzm2, HNS} have, in fact, been anticipated in the work of Olver and Sokolov \cite{Olver-Sokolov:IS-associative}.

A natural question is whether another multi-component CH is needed at all. First, we argue that calling \eqref{eq:2CH} another two-component CH equation is a misnomer, even though the second author used this term in his earlier publications \cite{BealsSzm1, BealsSzm2} on the subject. Instead, \eqref{eq:2CH} and \eqref{eq:VCliff line} should be viewed as the CH equation with an internal symmetry. This is particularly clear when one examines \eqref{eq:VCliff line}: the scalar CH equation represented by the field $u$ acquires internal degrees of freedom, controlled by the skew-symmetric tensor ${\bf V}$, which transforms covariantly under the orthogonal group $O(d-1)$. In the present case, $d=2$; ${\bf V}$ is a $2\times 2$ skew-symmetric matrix, represented by the coupling field $v$. The dynamics of the internal degrees of freedom are governed by the Clifford algebra structure. For more information on this point and the relevance of the Clifford algebra generalization of the 
Euler--Bernoulli beam, the reader might consult the recent papers \cite{BealsSzm2, HNS}.

Another aspect of the CH equation that takes an interesting turn concerns its particle nature.  Indeed, one of the most attractive features of the CH equation is that its $N$-peakon sector admits a very sharp wave-particle description: \eqref{eq:uwave} represents a wave with peaks at $x_j$, while equations \eqref{eq:CHpeakonsH} describe the dual mechanical system. The idea that solitons can be interpreted as particles dates back to Kruskal and Zabusky's landmark paper \cite{Kruskal-Zabusky}. This interpretation was given precise meaning within inverse scattering theory, where a finite-dimensional mechanical system can be identified using action-angle variables \cite{Faddeev-Zakharov-KdV}. However, the inverse canonical transformation to position space $x$ encounters structural obstacles, mainly because solitons are not sufficiently localized.
The situation with CH peakons provides a much sharper duality between waves and particles; the peak positions determine the particle positions uniquely, while the amplitudes of the individual peakons provide the 
momenta. Using spectral methods \cite{beals-sattinger-szmigielski:moment,beals-sattinger-szmigielski:Stieltjes}, it has been shown that, asymptotically in time, peakons become independent free particles.

In contrast to the scalar CH equation, where asymptotically separated peakons behave as independent particles, the present system exhibits \emph{persistent synchronization} of internal degrees of freedom. Even as the peak positions separate linearly in time, the amplitudes continue to exchange energy periodically. This produces a form of long-range coherence that persists despite asymptotic spatial decoupling. A somewhat similar behaviour has been observed in non-integrable cross-coupled CH equations studied by Cotter, Holm, Ivanov and Percival in \cite{Cotter_2011}, dubbed ``waltzing peakons'', but for bounded pairs rather than for spatially separated ones.  

The content of the paper can be summarized as follows:-
\begin{enumerate} 
\item In \autoref{sec:Lax}, we review the Clifford-based generalization of the Euler--Bernoulli beam.  The main purpose of this brief review is to place the current investigation in a broader context.  
We mostly follow \cite{BealsSzm2}, eventually returning to the main subject of our paper, namely the Clifford algebra with Minkowski signature and the system of equations \eqref{eq:2CH} and \eqref{eq:2CHconstraints}.  We conclude this section by stating the peakon equations for that system.  
\item In \autoref{sec:FP}, we study the forward spectral problem for the equation 
\[
D^2 \Phi=(\bo +\lambda M)\Phi 
\]
in the case of the peakon measure $M=\sum_{j=1}^N M_j \delta_{x_j}$.  In \autoref{thm:T}, we give an explicit formula for the transition matrix $T$, including the spectral invariants (constants of motion)
that appear as coefficients of the polynomials $\det(T_{11})$ and $\tr (T_{11} \sigma_1)$ (\autoref{cor:det/trace A}).  Finally, 
when $M$ is a positive measure (no anti-peakons), we establish in \autoref{lem:Lemma1} that peakons cannot collide.

\item The remainder of the paper is devoted to the case of two peakons.  The rationale is that, as in the soliton case, the two-peakon case 
is an avatar of the more complex dynamics involved in a multi-peakon case.  The dynamics of two peakons in \eqref{eq:CH} was already studied in the original paper of Camassa and Holm \cite{camassa-holm}, and one of the 
objectives of the present paper is to explore the ramifications of equipping \eqref{eq:CH}  with internal degrees by generalizing the CH equation to \eqref{eq:2CH} and \eqref{eq:2CHconstraints}.  
\eqref{sec:two peakons I} delves into the details of the two-peakon dynamics.  One important step is a precise description of the asymptotic tail of the coupling field $v$, denoted $v_+$.  
We identify a pair of canonical variables $(v_+,y)$ whose dynamics is governed by what we call a hyperbolic pendulum, and whose solution is expressed in terms of elliptic functions.  
\item In \autoref{sec:Dynamics}, we explore the asymptotic behaviour of the solutions by first establishing that, as time goes to infinity, peakons separate in space.  This is proven in 
\autoref{thm:gap}.  Furthermore,  the position-dependent factor $E:=e^{-2\abs{x_1-x_2}}$ that controls the separation is shown to be in $L^1(0,\infty)$, a fact we will exploit numerous times 
in further analysis.  In \autoref{subsec:Enull}, we elaborate on a simplified dynamics by setting $E\equiv 0$.  This leads us to define, in the space of masses (momenta), an asymptotic curve $\Gamma$ given in \eqref{eq:Gamma-curve}, which we uniformize using a hyperbolic parametrization (see \eqref{eq:hyper-sd}) and another parametrization in terms of $(v_+, y)$.  The latter is the subject of 
\autoref{prop:m1n1m2n2-d}.  

\item We subsequently analyze the full dynamics for $E \neq 0$, distinguishing two regimes: the boundary regime discussed in \autoref{sec:BC} and the generic regime treated in \autoref{sec:GC}. We demonstrate that the boundary regime reduces to the classical two-peakon problem for the original CH equation (see \autoref{sec:BC}). In contrast, the generic regime examined in \autoref{sec:GC} gives rise to a previously unobserved type of peakon interaction: the masses $(m_1, n_1, m_2, n_2)$ asymptotically approach a periodic orbit $\Gamma$ (see \autoref{prop:Gamma-approx}). In \autoref{sec:Dynamics of separation}, we derive an exact identity that connects the peak separation to the time-accumulated imbalance between the two amplitudes (see \autoref{thm:SI}). This identity, in conjunction with a frozen-parameter comparison principle, yields exponential decay of the interaction between the peakons and exponential convergence of the amplitude variables toward the limiting periodic orbit (see \autoref{cor:exponential-approach-Gamma}). Furthermore, the main result \autoref{thm:Tsampling} establishes, among other statements, that once the asymptotic regime has been entered, the separation grows from one period of the internal oscillation to the next, even though its instantaneous derivative may continue to change sign. A detailed description of the asymptotic trajectories is given in \autoref{prop:x-asymptotics}.

\item In \autoref{sec:gallery}, we present a collection of graphs that illustrate the principal dynamical differences between the CH peakons and the peakons with internal degrees of freedom analyzed in this work.

\item The paper concludes with appendices.  In \autoref{appendix:periodic}, we give the details about the periodic curve of masses $\Gamma$, one of the central concepts of the current paper.  
This is followed by \autoref{appendix:relations} in which we list a variety of interrelations between spectral invariants used in the body of the paper, while in \autoref{appendix:positions} we 
provide the computational details omitted in \autoref{sec:positions}.  

\end{enumerate}

\section{Lax pair formulation} \label{sec:Lax}
We will briefly recall a Clifford algebra-based generalization of \eqref{eq:2CH} and \eqref{eq:2CHconstraints} following \cite{BealsSzm2}.  
This generalization is not necessary for the remainder of this paper, but it provides a conceptual locale for the discussion of the significance of equations \eqref{eq:2CH} and \eqref{eq:2CHconstraints}.  

Let $W$ be a  $d$-dimensional real vector space with a bilinear, non-degenerate, symmetric form $B$, with signature $(p,q), p+q=d$. 
The Clifford algebra $\mathcal{C\ell} (W,B)$ is generated by relations 
\begin{equation} \label{eq:Clifdef}
w_1w_2+w_2w_1=2B(w_1,w_2), \qquad w_1, w_2 \in W. 
\end{equation}
In terms of a fixed orthonormal basis, the defining relations read
\begin{equation} \label{eq:Clifdef2}
e_\mu e_\nu+e_\nu e_\mu=2\epsilon_\mu \delta_{\mu \nu} , \qquad \mu, \nu  =1,2, \cdots, d, \qquad \epsilon_\mu=\pm 1.  
\end{equation} 
$\mathcal{C\ell} (W)$ splits as a vector space
\begin{equation} 
\mathcal{C\ell} (W)=\mathcal{C\ell}_0 (W)\oplus \mathcal{C\ell}_1 (W)\oplus \mathcal{C\ell}_2 (W)\oplus\cdots \oplus \mathcal{C\ell}_d(W). 
\end{equation} 
$\Cl_0 (W)$ is isomorphic to a copy of $\R$ while $\Cl_j(W)$ is spanned by $\binom{d}{j}$ elements $e_{\mu_1}e_{\mu_2}\cdots e_{\mu_j}$ for $\mu_1<\mu_2<\cdots<\mu_j$.  
$\Cl(W)$ is isomorphic as a vector space to the $2^d$ dimensional Grassmann algebra $\Lambda(W)=\bigoplus_{j=0}^d \Lambda ^j(W)$.  $\Cl_j(W)$ elements are assigned the (Grassmann) degree $j$.  
\begin{definition} \label{def:Mstring}
Consider a compactly supported measure $M\in \Cl_1(W)$.  
The boundary value problem 
\begin{align*} 
D_x^2 \Phi&=(1+\lambda M) \Phi, \qquad   -\infty<x<\infty, \,   \\
\Phi(x, \lambda) \rightarrow 0, \, \text{ as } x\rightarrow -\infty &,\, \text{ and $\Phi(x,\l)$ is non-invertible as } x\rightarrow \infty, 
\end{align*} 
will be called a \emph{Matrix String}.  
\end{definition} 
\begin{definition}\label{def:alphaMstring-Deformation}  
\begin{equation}\label{eq:alpha-t-flow-differential form} 
D_t \Phi=(\frac{a-b_x}{2} +b\, D_x)\Phi, 
\end{equation}
where $a,b\in \Cl(W) $ are functions of $x$, a deformation parameter $t$,  and the spectral variable $\lambda$, which are bounded as $\abs{x} \rightarrow \infty$.    
\end{definition} 
The generalization of \eqref{eq:2CH} written in vector form reads \cite{BealsSzm2}
\begin{equation} \label{eq:VCliff line}
\vec{M}_t=\mathcal{L}_{\vec{M}} u+AS\vec{M}, 
\end{equation} 
where $A=\begin{bmatrix} 0&\vec{v}\\-\vec{v}^T&0 \end{bmatrix}$, $S$ is the signature matrix, and 
$\mathcal{L}_{\vec{M}}u=(\vec{M} u)_x+u_x \vec{M}$.  Since writing the analog of \eqref{eq:2CHconstraints} requires more notation, we refer an interested reader to \cite{BealsSzm2}.  

The original Lax pair equation used in \cite{BealsSzm1}  to generate \eqref{eq:2CH} and \eqref{eq:2CHconstraints} corresponds to a Clifford algebra on two generators with the Minkowski metric $(+,-)$, and is given in an explicit matrix form by 
\begin{align}
 \Phi_{xx} &= (\bo+ \lambda M)\,\Phi,  \label{eq:xLax}\\
 \Phi_t &= \frac{a-b_x}{2}  \Phi+ b\,\Phi _x,  \label{eq:tLax} 
\end{align}
where
\begin{equation} \label{eq:Mab}
 M = 
 \begin{pmatrix}
   0 & n \\
   m & 0
 \end{pmatrix}, \quad 
 a = 
 \begin{pmatrix}
   v& 0 \\
   0   & -v 
 \end{pmatrix},  \quad
 b = 
 \begin{pmatrix}
   u & \frac1\lambda \\
   \frac1\lambda & u
 \end{pmatrix}, \qquad \bo=\begin{pmatrix}1&0\\0&1\end{pmatrix}, 
 \end{equation} 
 with  $ \lambda \in \mathbb{C}$ and scalars $m, n, u, v $.  
 
 For bounded $u$ and compactly supported $m,n$ one can express $u,v$ in terms of $m,n$ as 
\begin{subequations} 
\begin{align} 
u(x,t)=-\frac14 \int_{\R} e^{-2\abs{x-y}} (m(y,t)+n(y,t)) dy, \label{eq:u-integral} \\
v(x,t)=\frac12 \int_{\R} \sgn(x-y) (n(y,t)-m(y,t))\, dy, \label{eq:v-integral} 
\end{align} 
\end{subequations} 
thus eliminating the constraints given by \eqref{eq:2CHconstraints}.  
The asymptotic values of $u$ and $v$ play a significant role in the whole theory.  They can be easily extracted from the integral formulas above, giving 
\begin{equation} \label{eq:uvass}
\begin{matrix}\lim_{x\rightarrow \pm \infty} u(x,t)=0=\lim_{x\rightarrow \pm \infty}u_x(x,t), \\ 
\lim_{x\rightarrow \pm \infty}v(x,t)=\pm \frac12  \int_{-\infty}^\infty[n(y,t)-m(y,t)]\, dy\stackrel{\text{def}}{=}v_\pm(t). 
\end{matrix}\end{equation}
Since $M$ is compactly supported, we have 
\begin{align}  
&u=u_-(t)e^{2x}, \qquad &x<<0, \label{eq:asminus}\\
&u=u_+(t) e^{-2x}, \qquad &0<< x, \label{eq:asplus}
\end{align} 
where 
$$ 
u_\pm(t)=-\frac14 \int_{-\infty}^\infty e^{\pm 2y}[m(y,t)+n(y,t)] \, dy. 
$$
We note that 
\begin{align} 
&a=\frac{v_-}2 \sigma_3, &b&=u_-e^{2x} \bo +\frac1\l \sigma_1,\qquad  &x<< 0\label{eq:abminus} \\
&a=\frac{v_+}{2 } \sigma_3,  &b&=u_+e^{-2x} \bo +\frac1\l \sigma_1, \qquad & x>> 0, \label{eq:abplus} 
\end{align}
where $\sigma _1, \sigma_3$ are Pauli matrices $\begin{psmallmatrix} 0&1\\1&0 \end{psmallmatrix}, \begin{psmallmatrix}1&0\\0&-1 \end{psmallmatrix}$, respectively.
In this paper, we work with measure-valued data $(m,n)$.  In the discrete/\emph{peakon} case,  the measures are assumed to take the form: 
\begin{equation} \label{eq:mnpeakons}
m(t)=\sum_{j=1}^N m_j(t) \delta_{x_j(t)}, \qquad n(t)=\sum_{j=1}^N n_j(t) \delta_{x_j(t)}, 
\end{equation} 
i.e., both measures live on the same moving support $\{x_j(t)\}$.  With this ansatz the fields $u,v$ read: 
\begin{align} 
&u(x,t)=-\frac14 \sum_{j=1}^N (m_j(t)+n_j(t))  e^{-2\abs{x-x_j(t)}}, \\
&v(x,t)=\frac12 \big(\sum_{x_j(t)<x} (n_j(t)-m_j(t))-\sum_{x_j(t)>x} (n_j(t)-m_j(t))\big), \\
& v_+(t)=\frac12 \sum_{j=1}^N  (n_j(t)-m_j(t)).  \label{eq:v+}
\end{align} 

Evaluated in the sense of distributions, the PDEs \eqref{eq:2CH} and \eqref{eq:2CHconstraints}  reduce to a closed ODE system for $\{x_j, m_j, n_j\}$.  As in classical CH, the sites 
$x_j$ move with the flow: 
\begin{equation} \label{eq:dotxj}
\dot x_j=-u(x_j(t),t), 
\end{equation} 
while the masses $m_j, n_j$ evolve by first-order linear laws driven by local gradients of $u$ and the coupling $v$: 
\begin{equation} \label{eq:dotmjnj}
\dot m_j(t)=(u_x-v)\big |_{x=x_j(t)} m_j(t), \qquad \dot n_j(t)=(u_x+v)\big|_{x=x_j(t)} n_j(t), 
\end{equation} 
where the values of $u_x$ and $v$ at $x_j$ mean the arithmetic averages of the left and right limits.

\section{ Forward spectral problem; finite discrete measure} \label{sec:FP}
In this section, we study the forward problem for the equation 
$D^2 \Phi=(\bo +\l M) \Phi$ assuming that 
\begin{equation} \label{eq:M}
M=\sum_{j=1}^N \begin{pmatrix} 0&n_j\\m_j&0 \end{pmatrix} \delta_{x_j}:= \sum_{j=1}^N M_j \delta_{x_j}, 
\end{equation} 
and $\Phi(x)=e^x\bo, \quad x<<0$.  We have temporarily omitted the $t$ dependence in the formulas in this section.  
In each mass free segment $x_{j-1}<x<x_j, \,1\leq j\leq N+1$, 
$$ 
\Phi(x,\l)= e^x A_j(\l)+ e^{-x} B_j(\l), \qquad 
$$ 
where $A_1=\bo, B_1=0$ and, in the region $x>>0$ we write $A_{N+1}=A, \, B_{N+1}=B$.  
Since $\Phi$ satisfies $D^2 \Phi=(\bo +\l M) \Phi$, the continuity and jump conditions at $x_j$ imply
$$ 
\begin{bmatrix} A_{j+1}\\ B_{j+1} \end{bmatrix}=T_j \begin{bmatrix} A_j\\ B_j \end{bmatrix} , 
$$ 
where 
\begin{equation*}
T_j=\begin{bmatrix} \mathbf{1}+\frac \lambda 2 M_j&\frac \lambda 2 e^{{-}2x_j} M_j\\
-\frac\lambda 2 e^{2x_j}M_j & \mathbf{1}-\frac\lambda2  M_j   \end{bmatrix}. 
\end{equation*}
We note that 
$$ 
\begin{bmatrix} A\\ B \end{bmatrix}=T_N\cdots T_1 \begin{bmatrix} \bo \\ 0  \end{bmatrix}.  
$$

The first elementary observation can be formulated as follows.  
\begin{lemma} \label{lem: basic1}
Let $X_j=\begin{bmatrix} 1&e^{-2x_j}\\ -e^{2x_j}& -1 \end{bmatrix}$.  

Then 
\begin{enumerate} 
 
\item 
\begin{equation*}
X_j=\begin{bmatrix} 1\\-e^{2x_j} \end{bmatrix} \begin{bmatrix} 
1& e^{-2x_j} \end{bmatrix} 
\end{equation*} 
\item 
\begin{equation*}
T_j =\mathbf{1}\otimes \mathbf{1}+ \frac\lambda2 X_j\otimes M_j
\end{equation*}
\end{enumerate} 
where $\otimes$ is the tensor product (Kronecker product) of matrices.  
\end{lemma} 
 
 \begin{definition}
  The binomial coefficient $\binom{S}{j}$
  denotes the collection of $j$-element subsets
  of the set~$S$,
  and $[N]$ denotes the integer interval $\{ 1,2,3, \dots, N\}$.
  We always label the elements of a set
  $I \in \binom{[N]}{j}$ in increasing order:
  $I = \{ i_1 < i_2 < \dots < i_j \}$.
  Moreover, given a collection of matrices $\{a_1, a_2, \cdots, a_N\} $, we denote the ordered product of matrices labelled by the multi-index 
  $I$ as $a_I\stackrel{def}{=} a_{i_j}a_{i_{j-1}}\cdots a_{i_1}$.  
 
\end{definition}
\begin{lemma} 
Given a multi-index $I \in \binom{[N]}{j}$ and the set of 
matrices $\{X_1, X_2, \cdots, X_N\}$ as in \autoref{lem: basic1} then 
\begin{equation} 
X_I=\prod_{k=1}^{j-1} \big( 1-e^{-2(x_{i_{k+1}}-x_{i_k})}\big) 
\begin{bmatrix} 1 & e^{-2x_{i_1}}\\-e^{2x_{i_j}}& -e^{2(x_{i_j}-x_{i_1})}\end{bmatrix} 
\end{equation} 
with the proviso that the empty product is taken to be $1$ when $j=1$.  
\end{lemma} 
\begin{proof} 
It suffices to observe that 
\begin{equation*}
\begin{split} X_jX_i=
\begin{bmatrix} 1\\-e^{2x_j} \end{bmatrix} \begin{bmatrix} 
1& e^{-2x_j} \end{bmatrix} \begin{bmatrix} 1\\-e^{2x_i} \end{bmatrix} \begin{bmatrix} 
1& e^{-2x_i} \end{bmatrix}=\\(1-e^{-2(x_j-x_i)})\begin{bmatrix} 1\\-e^{2x_j} \end{bmatrix} \begin{bmatrix} 
1& e^{-2x_i} \end{bmatrix}, \qquad i<j, 
\end{split} 
\end{equation*} 
and then proceed by induction on the number of terms.  
\end{proof}

Let us define 
\begin{equation} \label{eq:fi}
\prod_{k=1}^{j-1} \big( 1-e^{-2(x_{i_{k+1}}-x_{i_k})}\big)\stackrel{def}{=}f_I. 
\end{equation}   
Then 
\begin{theorem} \label{thm:T} 
\begin{equation} 
T=\mathbf{I}\otimes \mathbf{I}+ 
\sum_{j=1}^N\big(\frac \lambda2\big) ^{j} \sum_{
I \in \binom{[N]}{j}}  f_I \begin{bmatrix} 1 & e^{-2x_{i_1}}\\-e^{2x_{i_j}}& -e^{2(x_{i_j}-x_{i_1})}\end{bmatrix} \otimes M_I
\end{equation}
In particular
\begin{align} 
T_{11}=A&=\mathbf{I}+\sum_{j=1}^N\big(\frac \lambda2 \big)^{j} \sum_{
I \in \binom{[N]}{j}}  f_I M_I, \\
T_{11}=A&=\big(\frac \lambda2\big)^{N} f_{[N]}M_{[N]} +O(\lambda^{N-1}), \qquad &\lambda \rightarrow \infty,  \\
\det(T_{11})=\det{A} &=\big(\frac \lambda 2\big)^{2N} f_{[N]}^2 \det M_{[N]}+O(\lambda^{2N-2}), \qquad &\lambda \rightarrow \infty, \label{eq:T11as} \\
T_{21}=B&=-\sum_{j=1}^N\big(\frac \lambda 2\big) ^{j} \sum_{
I \in \binom{[N]}{j}} e^{2x_{i_j}}f_I  M_I.  &
\end{align}
\end{theorem} 
%One obvious corollary is this. 
%\begin{corollary} 
%If none of the masses $\{m_1(0), \cdots, m_d(0); n_1(0), \cdots, n_d(0)\}$ 
%is zero, i.e. $m$ and $n$ have identical supports, then there are no collisions of masses, meaning, $x_i(t)\neq x_j(t), i\neq j$ for 
%all times.  
%\end{corollary} 
%\begin{proof} 
%It suffices to prove that the neighbours cannot collide.  We write $f_{[d]}=\prod_{k=1}^{d-1}\big( 1-e^{-2(x_{k+1}-x_{k})}\big)$, use 
%\autoref{thm:T}, \autoref{eq:T11as}, and observe that $\det M_{[d]}$ is bounded 
%for all times (this follows from equations of motion), and thus 
%$f_{[d]}$ cannot be zero.  Hence there are no collisions and thus $m$ and $n$ have identical supports for all time.  
%\end{proof} 
\begin{example} \label{example:N=2}
Set $N=2$.  Since $f_{1}=f_2=1$ we have only one non-trivial factor$f_{I}$, 
namely $ f_{[1,2]}$,   
Thus 
\begin{equation*}
A=\mathbf{I}+\frac \lambda 2 (M_1+M_2)+ \big(\frac \lambda 2\big)^2 \big( 1-e^{-2(x_2-x_{1})}\big)M_2M_1,
\end{equation*}
which can be written explicitly as 
\begin{center}
\begin{equation} \label{eq:A}
\boxed{A=\begin{bmatrix} 1+\big(\frac \lambda 2 \big)^2\big( 1-e^{-2(x_2-x_{1})}\big)n_2m_1& \frac \lambda 2(n_1+n_2)\\
\frac \lambda2 (m_1+m_2)& 1+\big(\frac \lambda 2\big)^2\big( 1-e^{-2(x_2-x_{1})}\big)m_2n_1\end{bmatrix}. 
}
\end{equation} 
\end{center} 

%
%
%We can also compute $B$: 
%\begin{center}
%\begin{equation} \label{eq:B} 
%\boxed{B=\begin{bmatrix} -\big(\frac \lambda 2\big)^2\big( 1-e^{-2(x_2-x_{1})}\big)e^{2x_2}n_2m_1&- \frac \lambda2 (e^{2x_1}n_1+e^{2x_2} n_2)\\
%-\frac \lambda 2 (e^{2x_1}m_1+e^{2x_2} m_2)& -\big(\frac \lambda 2\big) ^2 \big( 1-e^{-2(x_2-x_{1})}\big)e^{2x_2} m_2n_1\end{bmatrix}. 
%}
%\end{equation} 
%\end{center} 

\end{example}

\subsection{Time evolution} \label{sec:FP-time}
The most natural normalization near $x=-\infty$, outside of the support of $M$,  is to take $\Phi=e^{x}\bo$, as we did in the previous section.   However, this is incompatible with \eqref{eq:tLax} (see \cite{BealsSzm2} for further details).  Consequently, we are seeking a 
$t$ dependent gauge transformation $\Omega(t, \l)$ such that $\Psi=\Phi \Omega$ can satisfy 
both \eqref{eq:xLax} and \eqref{eq:tLax}.  Then, \eqref{eq:tLax} implies
\begin{equation} \label{eq:t-Omega} 
\Phi_t+\Phi \Omega_t \Omega^{-1} =\frac{a-b_x}{2}\Phi+b \Phi_x, 
\end{equation} 
and in the asymptotic region $x<< 0$ where $\Phi=e^x\bo $ we 
get 
\begin{equation} 
 \Omega_t \Omega^{-1} =\frac{a}{2}+b\stackrel{\eqref{eq:abminus}}{=} \frac{v_-}{2} \sigma_3 +\frac1\lambda \sigma_1.  \label{eq:Omegam}
\end{equation} 
The computation in the asymptotic region $x>> 0$ is only slightly more involved.  We recall that in that region 
$\Phi(x,t)=A(t,\l) e^x+B(t,\l) e^{-x}$.  
\begin{proposition} \label{prop:AB}
\begin{align} 
&A_t=J_+A+AJ_-, \\
&B_t=J_-B+BJ_- +2 u_+A, \\
&A_t \sigma_1=[J_+, A\sigma_1], \label{eq:ALax}
\end{align} 
where $J_\pm=\frac{v_+}{2}\sigma_3\pm \frac1\l \sigma_1$.    
\end{proposition} 
\begin{proof} 
Following an elementary computation based on the equation \eqref{eq:t-Omega}, one obtains the first two statements.  \eqref{eq:ALax} follows from the identity $J_-\sigma_1=-\sigma_1 J_+$.  
\end{proof} 

\begin{corollary} \label{cor:det/trace A}
The trace $\textrm{tr} A\sigma_1$ and determinant $\det{A}$ are constants of motion.  
Together, their coefficients provide 
$$ 
\# \textrm{constants}=\begin{cases} \frac{3N}{2} , &\textrm{ if N is even  } , \\  \frac{3N+1}{2}, &\textrm{ if N is odd. } \end{cases} 
$$.  
\end{corollary} 
\begin{proof} 

Recall the formula for $A$:
$$
A=\mathbf{I}+\sum_{j=1}^N\big(\frac \lambda2 \big)^{j} \sum_{
I \in \binom{[N]}{j}}  f_I M_I
$$
Let us separate the even (diagonal) terms from the odd (off diagonal) terms. We emphasize that the products below alternate between $m$'s and $n$'s, for each index in $I$. For example, for the set $I=\lbrace 1,3,7,10\rbrace$, $m_{i_j}n_{i_{j-1}}...n_{i_1}=m_{10}n_{7}m_{3}n_1$
\begin{align*}
    M_I=\begin{cases}
    \begin{bmatrix}
    n_{i_j}m_{i_{j-1}}...m_{i_1} & 0\\
    0 & m_{i_j}n_{i_{j-1}}...n_{i_1}
    \end{bmatrix}
    \hspace{0.9cm} \text{if} \hspace{0.15cm} |I| \hspace{0.15 cm}\text{is}\hspace{0.15cm}\text{even}, \\\\
    \begin{bmatrix}
    0 & n_{i_j}m_{i_{j-1}}...n_{i_1}\\
    m_{i_j}n_{i_{j-1}}...m_{i_1} & 0
    \end{bmatrix}
    \hspace{0.9cm} \text{if} \hspace{0.15cm} |I| \hspace{0.15cm}\text{is}\hspace{0.15cm}\text{odd}. 
    \end{cases}
\end{align*}
With this, the determinant becomes
\begin{align*}
    \det(A)&=\det(\mathbf{1}+\sum\limits_{\substack{ 2\leq j\leq N \\ j \hspace{0.1cm} \text{even}}} (\frac{\lambda}{2})^j\sum_{I\in \binom{[N]}{j}}f_I M_I + \sum\limits_{\substack{1\leq j\leq N \\ j \hspace{0.1cm} \text{odd}}} (\frac{\lambda}{2})^j\sum_{I\in \binom{[N]}{j}}f_I M_I)\\
    &=\det (\mathbf{1}+\sum\limits_{\substack{ 2\leq j\leq N \\ j \hspace{0.1cm} \text{even}}} (\frac{\lambda}{2})^j\sum_{I\in \binom{[N]}{j}}f_I M_I) +\det ( \sum\limits_{\substack{1\leq j\leq N \\ j \hspace{0.1cm} \text{odd}}} (\frac{\lambda}{2})^j\sum_{I\in \binom{[N]}{j}}f_I M_I)\\
    &=(1+\sum\limits_{\substack{2\leq j\leq N \\ j \hspace{0.1cm} \text{even}}} (\frac{\lambda}{2})^j\sum_{I\in \binom{[N]}{j}}f_I [n_{i_j}m_{i_{j-1}}...m_{i_1}])(1+\sum\limits_{\substack{2\leq k\leq N \\ k \hspace{0.1cm} \text{even}}} (\frac{\lambda}{2})^k\sum_{I\in \binom{[N]}{k}}f_I [m_{i_k}n_{i_{k-1}}...n_{i_1}])\\
    & -(\sum\limits_{\substack{ 1\leq j\leq N \\ j \hspace{0.1cm} \text{odd}}} (\frac{\lambda}{2})^j\sum_{I\in \binom{[N]}{j}}f_I [m_{i_j}n_{i_{j-1}}...m_{i_1}])(\sum\limits_{\substack{1\leq k\leq N \\ k \hspace{0.1cm} \text{odd}}} (\frac{\lambda}{2})^k\sum_{I\in \binom{[N]}{k}}f_I [n_{i_k}m_{i_{k-1}}...n_{i_1}]) \\
    \\
    &= 1+P_2\lambda^2+P_4\lambda^4+...+P_{2N}\lambda^{2N}
\end {align*}
Discarding the zero order term, the polynomial consists of the even powers of $\lambda$ from $\lambda^2$ up to $\lambda^{2N}$. This yields $N$ conserved coefficients for both even and odd $N$.
The trace $\mathrm{tr}(A\sigma_1)$, on the other hand, consists of odd powers of $\lambda$ less than or equal to $N$:
\begin{align*}
    \mathrm{tr}(A\sigma_1)&=(\sum\limits_{\substack{1\leq j\leq N \\ j \hspace{0.1cm} \text{odd}}} (\frac{\lambda}{2})^j\sum_{I\in \binom{[N]}{j}}f_I [n_{i_j}m_{i_{j-1}}...n_{i_1}])+(\sum\limits_{\substack{1\leq k\leq N \\ k \hspace{0.1cm} \text{odd}}} (\frac{\lambda}{2})^k\sum_{I\in \binom{[N]}{k}}f_I [m_{i_k}n_{i_{k-1}}...m_{i_1}])\\
    &=\begin{cases}
     Q_1\lambda+Q_3\lambda^3+...+Q_{N}\lambda^{N}
     \hspace{0.5cm} \text{if} \hspace{0.15cm} \text{N} \hspace{0.15cm}\text{is}\hspace{0.15cm}\text{odd}\\
     Q_1\lambda+Q_3\lambda^3+...+Q_{N-1}\lambda^{N-1}
     \hspace{0.5cm} \text{if} \hspace{0.15cm} \text{N} \hspace{0.15cm}\text{is}\hspace{0.15cm}\text{even}
    \end{cases}
\end{align*}
This yields $\frac{N}{2}$ conserved coefficients in the even case, and $\frac{N+1}{2}$ conserved coefficients in the odd case, which completes the proof.
\end{proof} 
\begin{example} 
Let us consider $N=2$.  The explicit form of $A$ was computed in \autoref{example:N=2}.  
Since $\tr{(A\sigma_1)}$ is a constant of motion, 
\begin{equation} \label{eq:I1} 
I_1=n_1+n_2+m_1+m_2, 
\end{equation}  
is a constant of motion as observed.

The determinant $\det(A)$ is another source of constants of motion.  
The outcome of a simple computation is
\begin{equation*} 
\det(A)=1-\big(n_1m_1+n_2m_2 +(m_1n_2+n_1m_2)e^{-2(x_2-x_1)} \big) \big(\frac \lambda2\big)^2+ \big(1-e^{-2(x_2-x_1)}\big)^2 m_2m_1n_2n_1 \big(\frac\lambda2\big)^4,  
\end{equation*} 
giving two additional constants of motion: 
\begin{equation} \label{eq:I2I3}
I_2=n_1m_1+n_2m_2 +(m_1n_2+n_1m_2)e^{-2(x_2-x_1)}  , \qquad I_3=\big(1-e^{-2(x_2-x_1)}\big)^2 m_2m_1n_2n_1.  
\end{equation}

%Let now review the equations of motion and the notation.  
%We have 
%\begin{subequations}\label{eq:CH2}
%\begin{align} 
%m_t=(um)_x+(u_x -v)m, & \qquad n_t=(un)_x+(u_x+v)n, \\
%u_{xx}-4u=m+n, & \qquad v_x=n-m. 
%\end{align} 
%\end{subequations}
%We use the following inversion formulas for $u$ and $v$ in terms of $m$ and $n$: 
%\begin{subequations} \label{eq:uv}
%\begin{align} 
%u(x,t)&=-\frac14 \int_\R e^{-2\abs{x-y}} [m(x,t)+n(x,t)] dy, \\
%v(x,t)&=\frac12 \int_\R \sgn(x-y) [n(y, t)-m(y, t)] dy. 
%\end{align} 
%\end{subequations} 
%In the case of discrete measures $m=\sum_{j=1}^N m_j(t)\delta_{x_j(t)}, \, n=\sum_{j=1}^N n_j(t)\delta_{x_j(t)}$, 
We conclude this section with a few general statements valid for measures $m,n$ supported on $N$ sites $x_j(t)$.  The explict form of the equation of motion follows from 
\eqref{eq:dotxj}, \eqref{eq:dotmjnj}, and reads:
\begin{subequations}
\begin{align} 
\dot x_j&=-u(x_j)=\frac14 \sum_{k=1}^N (m_k+n_k) e^{-2\abs{x_j-x_k}},\label{eq:dotx} \\
\dot m_j&=m_j (u_x(x_j)-v(x_j))=\frac12 m_j  \sum_{k=1}^N \sgn(x_j-x_k) \big[ (m_k+n_k)e^{-2\abs{x_j-x_k}}-(n_k-m_k)\big] , \label{eq:dotm}\\
\dot n_j&=n_j (u_x(x_j)+v(x_j))=\frac12 n_j  \sum_{k=1}^N \sgn(x_j-x_k) \big[ (m_k+n_k)e^{-2\abs{x_j-x_k}}+(n_k-m_k)\big].  \label{eq:dotn} 
\end{align} 
\end{subequations} 
The following two constants of motion are essential to the preliminary analysis; their status as constants of motion can be verified directly from the equations of motion and also follows from \autoref{thm:T}.  We denote by $I_{Top} $ the highest coefficient in 
$A$, computed in  \eqref{eq:T11as}.  
\begin{proposition} 

\begin{equation} \label{eq:I1IT} 
I_1=\sum_{j=1}^N (m_j+n_j), \qquad I_{Top} =\prod_{j=1}^Nm_jn_j \prod_{k=1}^{N-1} (1-e^{-2(x_{k+1}-x_k)})^2  
\end{equation} 
are constant.  
\end{proposition} 
\begin{lemma} \label{lem:Lemma1} 
\mbox{}
Assume all masses $m_j(0), n_j(0), 1\leq j  \leq N,$ are positive and all positions are distinct $x_1(0)<x_2(0)<\cdots< x_N(0)$.  Then 
\begin{enumerate} 
\item Masses (momenta) stay positive. 
\item Masses $m_j, n_j$ are bounded above by $I_1$ and all are bounded away from zero.  
\item Peakons move to the right. 
\item Peakons cannot collide (i.e., for no time $t$, $x_k(t)= x_{k+1}(t)$).  In particular, they remain ordered $x_1(t)<x_2(t)<\cdots < x_N(t)$.   
\end{enumerate} 
\end{lemma} 
\begin{proof} 
By \eqref{eq:dotm} and \eqref{eq:dotn}, the masses cannot change sign.  Moreover, since the sum of masses equals $I_1$, each $m_j$ and $ n_j$ is bounded from above by $I_1$.  Likewise, 
the masses are bounded away from $0$.  Indeed, suppose there exists a sequence $t_n$ for which, say, $m_j(t_n)$ 
converges to $0$ (zero is an accumulation point).  Then, at least one $m_k(t_n)$ or $n_k(t_n)$ would have to converge to $\infty$ to respect the constancy of $I_{Top}$.  Since the masses are bounded from above, this cannot happen.  
Since the masses are bounded both from above and below, there exist $0<A$ and $0<B$ such that $0<A<\dot x_j(t)<B$, which proves that peakons move to the right.  
Moreover, by using the same argument as the one used for the boundedness of masses, but this time applied to positions, if at some time $t$, 
$x_{k}(t)=x_{k+1}(t)$, then, to preserve $I_{Top}$, one or more 
masses would have to become unbounded, which is prohibited.  
\end{proof}

\end{example} 
 \section{Two peakons; preliminaries} \label{sec:two peakons I}
 The remainder of the paper addresses the dynamics of two peakons in the Clifford-based extension of the CH equation.  In the Conclusions, we will make a qualitative comparison of the 
 dynamics of two peakons in the CH equation with the current work.  
 
 \subsection{ Peakon equations for \texorpdfstring{$N=2$}{N=2}} 
 
The dynamical system we are considering consists of six equations
\begin{subequations} \label{eq:DS}
\begin{align} 
\dot x_1&=\frac14\big[(m_1+n_1)+(m_2+n_2) e^{-2\abs{x_1-x_2}}\big], \label{eq:x1}\\
\dot x_2&=\frac14\big[(m_1+n_1)e^{-2\abs{x_1-x_2}}+(m_2+n_2) \big], \label{eq:x2}\\
\dot n_1&=\frac12 n_1\big[ -(m_2+n_2)e^{-2\abs{x_1-x_2}}-(n_2-m_2) \big], \label{eq:n1}\\
\dot m_1&=\frac12 m_1\big[ -(m_2+n_2)e^{-2\abs{x_1-x_2}}+(n_2-m_2) \big], \label{eq:m1}\\
\dot n_2&=\frac12 n_2\big[ (m_1+n_1)e^{-2\abs{x_1-x_2}}+(n_1-m_1)\big], \label{eq:n2}\\
\dot m_2&=\frac12 m_2\big[ (m_1+n_1)e^{-2\abs{x_1-x_2}}-(n_1-m_1) \label{eq:m2} \big],   
\end{align} 
\end{subequations} 
for variables $x_1,x_2, m_1, m_2, n_1, n_2$ subject to the conditions $m_j(0)>0, n_j(0)>0$ and the ordering condition $x_1(0)<x_2(0)$.  By \autoref{lem:Lemma1}, these conditions are preserved under the evolution governed by 
\eqref{eq:DS}.  

 We observe that the evolution of positions is formally the same as in the CH case for amplitudes $m_j+n_j$. This observation remains valid for all $N$. However, the evolution of the masses (or momenta in the CH theory) depends both on $m_j+n_j$ and on $m_j-n_j$, making it significantly different from the CH theory.  

We already identified in the previous section three constants of motion, $I_1, I_2, I_3$.  We recall them for convenience: 
\begin{subequations} 
\begin{align}
I_1&=m_1+n_1+m_2+n_2,\label{eq:I1bis} \\
 I_2&=n_1m_1+n_2m_2 +(m_1n_2+n_1m_2)e^{-2(x_2-x_1)}, \label{eq:I2} \\
 I_3&=\big(1-e^{-2(x_2-x_1)}\big)^2 m_2m_1n_2n_1.   \label{eq:I3}
\end{align}
\end{subequations}

\subsection{ Lax evolution; explicit computation  of \texorpdfstring{$v_+(t)$ for $N=2$}{v+ for N=2} }\label{sec:pendulum} 
Recall that for $N=2$
\begin{equation} 
A=\begin{bmatrix} 1+(\frac \lambda2)^2{\big( 1-e^{-2(x_2-x_{1})}\big)}n_2m_1& \frac \lambda2 (n_1+n_2)\\
\frac \lambda2 (m_1+m_2)& 1+(\frac \lambda2)^2{ \big( 1-e^{-2(x_2-x_{1})}\big)}m_2n_1\end{bmatrix}
\end{equation} 
We  introduce a particular parameterization for the $\lambda^2$ terms 
in $A$, taking advantage of the fact that the product of these terms equals $I_3$ (see \eqref{eq:I3}), by setting 
\begin{equation}\label{eq:def y}
{\big( 1-e^{-2(x_2-x_{1})}\big)}n_2m_1=\sqrt{I_3} e^y, \qquad 
{\big( 1-e^{-2(x_2-x_{1})}\big)}m_2n_1=\sqrt{I_3} e^{-y}
\end{equation}

Moreover, we normalize $A\sigma_1$ by defining $S=A\sigma_1-\frac12 \tr{A\sigma_1} \mathbf{1}$, 
which in the above parametrization reads: 
\begin{equation} \label{eq:S2} 
S=\begin{bmatrix}\frac \lambda2  v_+& 1+(\frac \lambda2)^2\sqrt{I_3} e^y\\
1+(\frac \lambda2)^2\sqrt{I_3} e^{-y}& -\frac \lambda2 v_+\end{bmatrix}, 
\end{equation}
where $v_+$ was defined in \eqref{eq:v+}.  
Note that $S$ satisfies the same Lax equation as $A \sigma_1$ (see \autoref{prop:AB}), 
namely, $\dot S=[J_+,S]$, and thus we obtain the following proposition. 

\begin{proposition} \label{prop:evolv+} 
The pair $(y,v_+)$ satisfies
\begin{equation} \label{eq:dotyv}
\dot y=v_+,\qquad \dot v_+=-\sqrt{I_3}\,\sinh y,
\end{equation} 
with Hamiltonian
\begin{equation} \label{eq:H}
H(y,v_+)=\tfrac12 v_+^2+\sqrt{I_3}\,\cosh y.  
\end{equation}
\end{proposition} 
\begin{proof} 
Computing the Lax equation, we get: 
\begin{align} 
\dot S_{11}=\frac{S_{21}-S_{12}}{\l}, \\
\dot S_{12}=v_+ S_{12}-\frac{2 S_{11}}{\l}, \\
\dot S_{21}=-v_+S_{21}+\frac{2 S_{11}}{\l}.  
\end{align} 

Using the parametrization from \eqref{eq:S2} we obtain the claim.  
\end{proof}

We immediately have 
\begin{lemma} \label{lem:ydoubledot} 
 $y$ satisfies 
 
\begin{equation} \label{eq:hpendulum}
\ddot y+\sqrt{I_3} \sinh y=0
\end{equation} 
\end{lemma} 
%\todo{ I'm not sure we should include the remark below. We'll keep it for now, but may remove it later.  } 
\begin{remark} 
We recall that the pendulum of length $l$ 
satisfies 
\begin{equation} 
\ddot \theta +\frac g l \sin \theta =0, 
\end{equation} 
so for this reason we refer to  \eqref{eq:hpendulum}  as a `` hyperbolic pendulum''.  
Since the solution of the pendulum can be expressed in terms of 
the Jacobi sine function $\sn$, we expect that the solution of the hyperbolic case will also be given in terms of the elliptic functions.  \end{remark}

\subsection{Hamiltonian system and energy normalization} \label{sec:hyperpendulum}
%\todo{ Perhaps, this should be an Appendix? } 
We now give a self-contained solution to the system \eqref{eq:dotyv}.  
Recall that we consider the Hamiltonian system
\begin{equation}\label{eq:sys}
\dot y=v_+(t),\qquad \dot v_+(t)=-\asqrt\,\sinh y(t),
\end{equation}
with Hamiltonian
\begin{equation*}
H(y,v_+)=\frac12 v_+^2+\asqrt\,\cosh y.
\end{equation*}
Hence the ``energy''
\begin{equation}\label{eq:Edef}
\mathcal E=\frac12 (\dot y)^2+\asqrt\,\cosh y
\end{equation}
remains constant along solutions.
Since $\cosh y\ge 1$, we have the bound $\mathcal E\ge \asqrt$, and we choose to work with the following parametrization: 
\begin{equation}\label{eq:Enorm}
\frac12 (\dot y)^2+\asqrt\,\cosh y=\asqrt\,(1+2k^2),\qquad k\ge 0.
\end{equation}

\subsubsection*{Step 1: Reduction to elliptic integrals}

Let $a:=\asqrt>0$. From \eqref{eq:Enorm} we get
\begin{equation}\label{eq:ydot2}
(\dot y)^2=2a\bigl(1+2k^2-\cosh y\bigr).
\end{equation}
Let us use the identity
\begin{equation}\label{eq:coshid}
\cosh y=1+2\sinh^2\Bigl(\frac{y}{2}\Bigr),
\end{equation}
and introduce the half--variable
\begin{equation}\label{eq:u}
 u:=\frac{y}{2}.
\end{equation}
Then $\dot y=2\dot u$, and \eqref{eq:ydot2} becomes
\begin{align}
4\dot u^2
&=2a\Bigl(1+2k^2-\bigl(1+2\sinh^2 u\bigr)\Bigr)
=4a\bigl(k^2-\sinh^2 u\bigr),
\end{align}
that is,
\begin{equation}\label{eq:udot2}
\dot u^2=a\bigl(k^2-\sinh^2 u\bigr).
\end{equation}
On any interval where $|\sinh u|\le k$ (which is automatic along a trajectory at energy level \eqref{eq:Enorm} (see \eqref{eq:udot2})), we may set
\begin{equation}\label{eq:subs}
\sinh u = k\sin\phi.
\end{equation}
Then
\begin{equation}\label{eq:cosh}
\cosh u=\sqrt{1+\sinh^2 u}=\sqrt{1+k^2\sin^2\phi}.
\end{equation}
Differentiating \eqref{eq:subs} gives
\begin{equation}\label{eq:udot}
\dot u=\frac{k\cos\phi}{\sqrt{1+k^2\sin^2\phi}}\,\dot\phi.
\end{equation}
On the other hand, \eqref{eq:udot2} yields
\begin{equation}\label{eq:udotE}
\dot u=\pm\sqrt{a}\,\sqrt{k^2-k^2\sin^2\phi}=\pm\sqrt{a}\,k\cos\phi.
\end{equation}
Combining \eqref{eq:udot} and \eqref{eq:udotE} (and cancelling the factor $k\cos\phi$ when it is nonzero), we obtain
\begin{equation}\label{eq:phidot}
\dot\phi=\pm\sqrt{a}\,\sqrt{1+k^2\sin^2\phi}.
\end{equation}
Thus
\begin{equation}\label{eq:ellint}
\int^{\phi(t)}\frac{d\theta}{\sqrt{1+k^2\sin^2\theta}}
=\pm\sqrt{a}\,(t-t_0).
\end{equation}
Recognizing this as an elliptic integral of the first kind with 
\emph{negative parameter} $m=-k^2$ \cite{AS64}[Section 16.10], namely
\begin{equation}\label{eq:Fdef}
F(\phi\mid m)=\int_0^{\phi}\frac{d\theta}{\sqrt{1-m\sin^2\theta}},
\qquad m=-k^2,
\end{equation}
we may write
\begin{equation}\label{eq:F}
F\bigl(\phi(t)\mid -k^2\bigr)=\pm\sqrt{a}\,(t-t_0).
\end{equation}

\subsubsection*{Step 2: Closed form solution via Jacobi elliptic functions}

Let $\mathrm{sn}(\cdot\mid m)$, $\mathrm{cn}(\cdot\mid m)$ denote the Jacobi elliptic functions with parameter $m$.
Taking the inverse of \eqref{eq:F} via the Jacobi amplitude, we have
\begin{equation}\label{eq:phi}
\sin\phi(t)=\mathrm{sn}\bigl(\pm\sqrt{a}\,(t-t_0)\mid -k^2\bigr).
\end{equation}
Substituting into \eqref{eq:subs} and recalling $u=y/2$, we obtain the explicit solution
\begin{equation}\label{eq:ysol}
\boxed{\;
 y(t)=2\,\arcsinh\Bigl(k\,\mathrm{sn}\bigl(\omega (t-t_0)\mid -k^2\bigr)\Bigr),
\qquad \omega=\Ithree^{1/4}.
\;}
\end{equation}

Differentiating \eqref{eq:ysol} and using 
\(
\frac{d}{dz}\mathrm{sn}(z\mid m)=\mathrm{cn}(z\mid m)\mathrm{dn}(z\mid m)
\)
along with 
\(
\mathrm{dn}^2(z\mid m)=1-m\,\mathrm{sn}^2(z\mid m)
\)
(and $m=-k^2$), one obtains a cancellation which yields the compact expression
\begin{equation}\label{eq:vsol}
\boxed{\;
 v_+(t)=\dot y(t)=2k\,\omega\,\mathrm{cn}\bigl(\omega (t-t_0)\mid -k^2\bigr).
\;}
\end{equation}
One checks directly that \eqref{eq:ysol}--\eqref{eq:vsol} satisfy the system \eqref{eq:sys} and conserve the energy level \eqref{eq:Enorm}.

\subsubsection*{Equilibrium case}
If $k=0$, then \eqref{eq:Enorm} forces $\mathcal E=\asqrt$, which occurs only at
\(y\equiv 0\) and \(v_+\equiv 0\).
Thus the equilibrium is the 
\emph{only} non-periodic trajectory.

\subsubsection*{Non-equlibrium case: period of motion}
For $k>0$, both $\mathrm{sn}(\cdot\mid -k^2)$ and $\mathrm{cn}(\cdot\mid -k^2)$ are periodic with real period $4K(-k^2)$, where
\begin{equation}
K(m)=F\Bigl(\frac{\pi}{2}\mid m\Bigr)
\end{equation}
is the complete elliptic integral of the first kind.
Therefore the solution \eqref{eq:ysol}--\eqref{eq:vsol} has period
\begin{equation}\label{eq:period}
T=\frac{4}{\omega}\,K(-k^2)=\frac{4}{\Ithree^{1/4}}\,K(-k^2),\qquad (k>0).
\end{equation}

\subsubsection*{Representative plots}

The plots below illustrate typical periodic solutions for $k>0$, and the equilibrium solution for $k=0$.
The figures were generated from \eqref{eq:ysol}--\eqref{eq:vsol} with $\Ithree=1$ (so $\omega=1$) and $t_0=0$.

\begin{figure}[h!]
\centering
\begin{subfigure}[t]{0.48\textwidth}
  \centering
  \includegraphics[width=\textwidth]{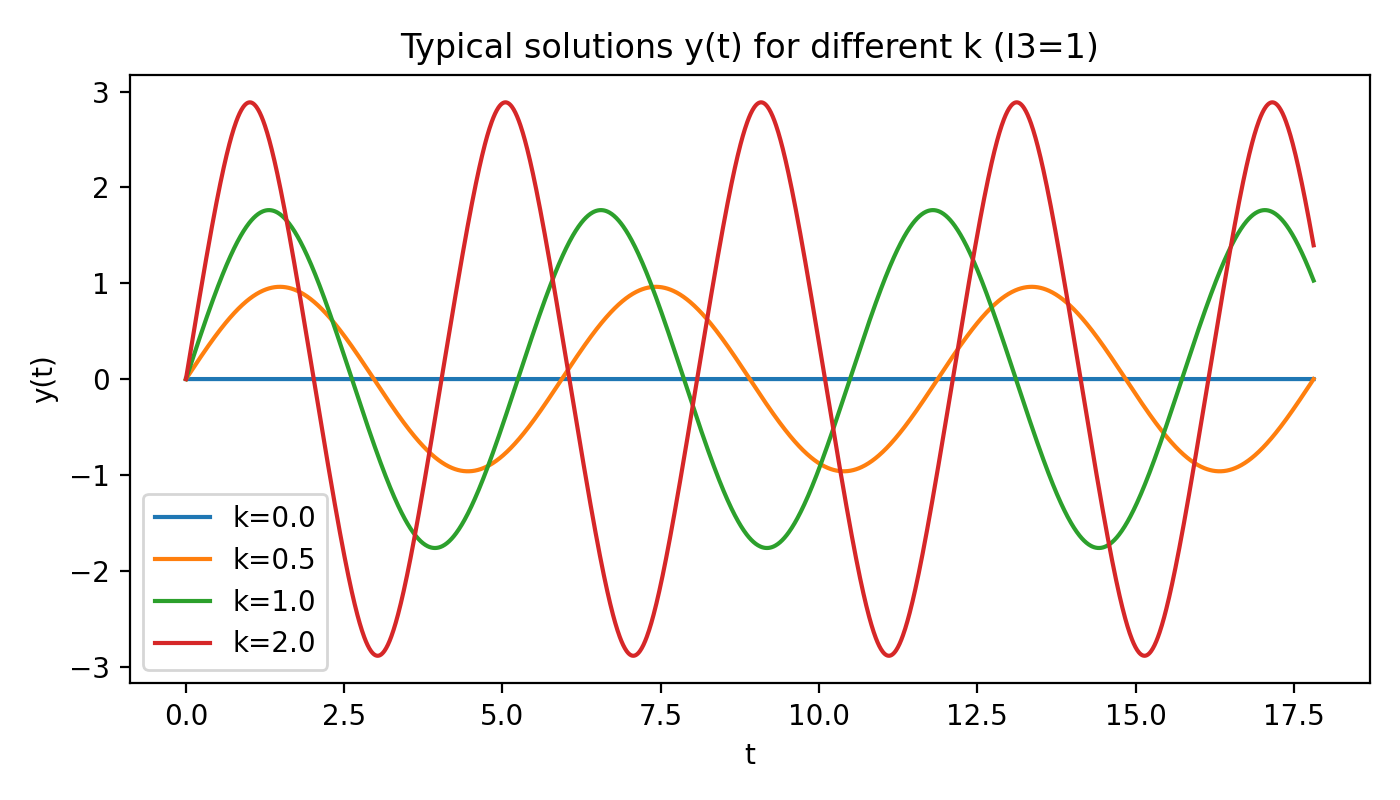}
  \caption{$y(t)=2\,\operatorname{arsinh}\bigl(k\,\mathrm{sn}(\omega t\mid -k^2)\bigr)$}
\end{subfigure}
\hfill
\begin{subfigure}[t]{0.48\textwidth}
  \centering
  \includegraphics[width=\textwidth]{KS_im/v_plot.png}
  \caption{$v_+(t)=2k\,\omega\,\mathrm{cn}(\omega t\mid -k^2)$}
\end{subfigure}
\caption{Typical solutions for several values of $k>0$ are periodic; for $k=0$ the solution is the equilibrium $(y,v_+)\equiv(0,0)$.}
\end{figure}
\begin{remark} Since $v_+=\frac12 (n_1-m_1+n_2-m_2)$ is now known, with the help of the conserved quantity $I_1$ (see \eqref{eq:I1}), 
 we can compute $n_1+n_2$, as well as $m_1+m_2$.  Similarly, because we know $y$, we can compute 
$(1-e^{-2(x_2-x_1)} )m_1n_2$ and $(1-e^{-2(x_2-x_1)})n_1m_2$, using $I_3$ (see  \eqref{eq:I3} and \eqref{eq:def y}).  
\end{remark} 

\section{Two peakons; separation and the limiting system} \label{sec:Dynamics} 
\subsection{Asymptotic separation in space}
{\bf Notation}: 
With two peakons labelled $1$ (left) and $2$ (right) we write
$$
x_1 < x_2,\qquad D := x_2-x_1>0,\qquad
E := e^{-2D}\in(0,1).
$$
Momenta and ``masses'' are denoted by 
$$
m_k(t)>0,\qquad n_k(t)>0,\qquad
s_k := m_k+n_k>0,\qquad
R := s_2-s_1, \quad k=1,2.  
$$
Differences between the two species are denoted by
$$
d_k := m_k-n_k,\qquad k=1,2.  
$$
%\subsection*{Field values at the peak positions and equations of motion}
%
%We record the field ($u$ and $v$) values at the peaks $ x_1$  and $ x_2$.  
%\begin{align*}
%u(x_1) &= -\tfrac14\bigl[s_1 + s_2\,E\bigr],
%&
%u(x_2) &= -\tfrac14\bigl[s_2 + s_1\,E\bigr],\\[4pt]
%u_x(x_1) &= \phantom-\,\tfrac12\,s_2\,E,
%&
%u_x(x_2) &= -\,\tfrac12\,s_1\,E,\\[4pt]
%v(x_1) &= -\,\tfrac12\,(n_2-m_2) = \tfrac12d_2,
%&
%v(x_2) &= \phantom-\,\tfrac12\,(n_1-m_1)=-\phantom-\,\tfrac12 d_1.
%\end{align*}
%
%\begin{subequations} 
%\begin{align} 
%&\dot x_1=\frac14(s_1+s_2 E), 
%&&\dot x_2=\frac14(s_1 E+s_2), \label{eq:x1x2}\\
%&\dot m_1=\frac{m_1}{2} (-s_2 E-d_2), 
%&&\dot n_1=\frac{n_1}{2}(-s_2E+d_2), \label{eq:m1n1}\\
%&\dot m_2=\frac{m_2}{2}(s_1 E+d_1), 
%&&\dot n_2=\frac{n_2}{2} (s_1 E-d_1).  \label{eq:m2n2}
%\end{align} 
%\end{subequations} 
%Equivalently,

We also note that by equations \eqref{eq:DS} we have 
\begin{equation} \label{eq:dotD}
\dot D = \tfrac14\,R\,(1-E), 
\end{equation} 
as well as
\begin{align}
\dot s_1 &= -\tfrac12 s_2E\,s_1 - \tfrac12 d_1 d_2, &
\dot d_1 &= -\tfrac12 s_2E\,d_1 - \tfrac12 d_2 s_1, \label{eq:SD1}\\
\dot s_2 &= +\tfrac12 s_1E\,s_2 + \tfrac12 d_1 d_2, &
\dot d_2 &= +\tfrac12 s_1E\,d_2 + \tfrac12 d_1 s_2. \label{eq:SD2}
\end{align}
We are now ready to prove that peakons \emph{scatter}, i.e., the distance $D(t)\rightarrow \infty$ as $t\rightarrow \infty$.  
\begin{theorem} \label{thm:gap}
Assume all masses $m_1(0),n_1(0), m_2(0), n_2(0)$ are positive and  $x_1(0)<x_2(0)$.  Then $E(t)\in L^1(0,\infty)$ and the distance between peakons $\lim_{t\rightarrow \infty} D(t)=\infty$.  
\end{theorem} 
\begin{proof} 
We integrate \eqref{eq:m2} and \eqref{eq:n2} to get 
\begin{equation*} 
m_2(t)n_2(t)=m_2(0)n_2(0) e^{\int_0^t s_1(t')E(t') dt'}.
\end{equation*} 
Since masses are bounded and $s_1(t')E(t')>0$ the integral $\int_0^\infty s_1(t)E(t) dt<\infty$.  By  \autoref{lem:Lemma1}, $s_1$ is bounded away from $0$.  
Hence, $\int_0^\infty E(t) dt<\infty$.  We claim that $\lim_{t\rightarrow \infty } E(t)$ exists and must be $0$.  To prove that, we first observe 
that by \eqref{eq:dotD} 
\begin{equation*} 
\abs{\dot E}\leq \frac{I_1}{2}.  
\end{equation*} 
Suppose $\lim_{t\rightarrow} E(t)$ is not $0$.  Then,  there must exist a number $L>0$, and a sequence $\{t_j\}$ converging to $\infty$ such that 
$\lim_{j\rightarrow \infty} E(t_j)=L$.  Let us take $\epsilon>0$, which is sufficiently small so that 
$L-\epsilon>0$.  We can always assume that the sequence $t_j$ is increasing and $t_{j+1}-t_{j}> \frac{2(L-\epsilon)}{I_1}$.  For sufficiently large $n$, $E(t_n)>L-\epsilon$. 
Consider  now $t\in [t_n, t_{n}+\frac{(L-\epsilon)}{I_1}]$.  Then, by the Mean Value Theorem, 
\begin{equation} 
E(t)>L-\epsilon-\abs{\dot E(\xi)}\frac{(L-\epsilon)}{I_1}>L-\epsilon-\frac{I_1}{2} \frac{(L-\epsilon)}{I_1}=\frac{L-\epsilon}{2} >0.  
\end{equation} 
Thus 
\begin{equation*} 
\int_0^\infty E(t) dt\geq \sum_{j} \frac{(L-\epsilon)^2}{2I_1}=\infty, 
\end{equation*}
hence a contradiction.  
\end{proof} 

\begin{remark}[Alternative one-line argument via Barbalat's lemma]
From \eqref{eq:m2} and \eqref{eq:n2} we already have $\int_0^\infty s_1(t)E(t)\,dt<\infty$.
By ~\autoref{lem:Lemma1}, $s_1(t)\ge c_*>0$, hence $\int_0^\infty E(t)\,dt<\infty$.
Moreover, from \eqref{eq:dotD} we have $|\dot E(t)|=\tfrac12 E(t)(1-E(t))\,|R(t)|\le I_1/8$, so $E$ is uniformly continuous.
By \emph{Barbalat's lemma} (if $f\in L^1([0,\infty))$ and $f$ is uniformly continuous, then $f(t)\to 0$ as $t\to\infty$), we conclude
$E(t)\to 0$. Hence $D(t)=-\tfrac12\log E(t)\to \infty$. \footnote{See, e.g., H.\ K.\ Khalil, \emph{Nonlinear Systems}, 3rd ed., Prentice Hall, 2002, Lemma~4.3; or J.-J.\ E.\ Slotine and W.\ Li, \emph{Applied Nonlinear Control}, Prentice Hall, 1991, App.~A.}
\end{remark}
\begin{corollary} 
The limits $\lim _{t\rightarrow \infty} m_1(t)n_1(t)$ and \, $\lim_{t\rightarrow \infty} m_2(t)n_2(t)$ exist and are two positive numbers in the interval $(0, I_1^2)$.  
\end{corollary} 
For future use, we will set 
\begin{align} s_1(t)^2-d_1(t)^2=&4m_1(t)n_1(t)=K_1(t), \quad &s_2(t)^2-d_2(t)^2=4m_2(t)n_2(t)=K_2(t),\label{eq:K1K2} \\
\lim _{t\rightarrow \infty} &4m_1(t)n_1(t)=K_1^\infty, \quad &\lim_{t\rightarrow \infty} 4m_2(t)n_2(t)=K_2^\infty. \label{eq:K1K2infty}
\end{align} 
%\subsection*{Reduced ODE system}

\subsection{Asymptotic decoupling \texorpdfstring{($\, E \rightarrow 0$)}{E ----> 0}} \label{subsec:Enull}
We are interested in the asymptotic limit where $D(t) \rightarrow \infty$ as $t \rightarrow \infty$. To investigate this, we will consider a simplified dynamical system by setting in the previous system $E\equiv 0$ and writing equations for the quadruple $(s_1, d_1, s_2, d_2)$.   We obtain 
\begin{equation} \label{eq:tD}
\dot D = \tfrac14\,R,  
\end{equation} 
and 
\begin{subequations} 
\begin{align}
\dot s_1&=-\tfrac12 d_1 d_2, 
&&\dot d_1 = -\;\tfrac12\,s_1\,d_2, \label{eq:ts1d1}\\[4pt]
\dot s_2&=\tfrac12 d_1 d_2,  
&&\dot d_2 = \tfrac12\,s_2\,d_1, & \label{eq:ts2d2}
\end{align}
\end{subequations} 
and note that $s_1+s_2=I_1$ and $K_j=s_j^2-d_j^2>0$ are now exact constants.  
Thus, for fixed $I_1,K_1, K_2$, the solution of the system   \eqref{eq:ts1d1} and \eqref{eq:ts2d2} lies on the closed curve defined by 

\begin{equation} \label{eq:Gamma-curve}
\Gamma=\{(s_1,d_1, s_2, d_2): s_1+s_2=I_1, \quad  s_1^2-d_1^2=K_1, \quad  s_2^2-d_2^2=K_2, \quad s_j>0, \quad K_j>0\}.  
\end{equation} 
We note that this curve degenerates to a point when $\sqrt{K_1}+\sqrt{K_2}=I_1$.  We will return to this degenerate case in \autoref{sec:BC}.  For now, let us introduce a hyperbolic parametrization with angles $(\theta_1, \theta_2)$: 
\begin{equation} \label{eq:hyper-sd}
s_j=\sqrt{K_j} \cosh \theta_j, \qquad d_j=\sqrt{K_j} \sinh \theta_j, \qquad j=1,2.  
\end{equation} 
Then \eqref{eq:ts1d1}--\eqref{eq:ts2d2} are equivalent to
\begin{equation}
\dot\theta_1=-\tfrac12\sqrt{K_2}\,\sinh\theta_2,\qquad
\dot\theta_2=\tfrac12\sqrt{K_1}\,\sinh\theta_1, \label{eq:theta_system}
\end{equation}
subject to the \emph{compact} constraint (from $I_1=s_1+s_2$).  
We see that the degenerate case $\sqrt{K_1}+\sqrt{K_2}=I_1$ happens precisely at the equilibrium point ($\theta_1=\theta_2=0$) of the flow since 
\begin{equation}
\sqrt{K_1}\cosh\theta_1+\sqrt{K_2}\cosh\theta_2=I_1=\text{const}. \label{eq:compact_constraint}
\end{equation}

Because $\cosh\ge 1$, \eqref{eq:compact_constraint} confines $(\theta_1,\theta_2)$ to a compact rectangle
\[
\cosh\theta_1\le \frac{I_1-\sqrt{K_2}}{\sqrt{K_1}},\qquad
\cosh\theta_2\le \frac{I_1-\sqrt{K_1}}{\sqrt{K_2}},
\]
so all variables remain bounded.

Differentiating the first equation in \eqref{eq:theta_system} and eliminating $\theta_2$ using \eqref{eq:compact_constraint} yields a one-degree-of-freedom conservative equation
\begin{equation}
\ddot\theta_1+\frac{\sqrt{K_1}}{4}\Big(I_1-\sqrt{K_1}\cosh\theta_1\Big)\sinh\theta_1=0. \label{eq:theta1_ode}
\end{equation}
Hence the ``energy''
\begin{equation}
\hat{\mathcal{E}}(\theta_1,\dot\theta_1)=\tfrac12\dot\theta_1^2+U(\theta_1),\qquad
U'(\theta_1)=\frac{\sqrt{K_1}}{4}\Big(I_1-\sqrt{K_1}\cosh\theta_1\Big)\sinh\theta_1, \label{eq:energy}
\end{equation}
is conserved.   Analysis for $\theta_2$ is similar by symmetry.

Because \eqref{eq:compact_constraint} confines $(\theta_1,\theta_2)$ to a compact invariant curve and the vector field \eqref{eq:theta_system} has no equilibrium on this curve under the strict inequality
\[
\sqrt{K_1}+\sqrt{K_2}<I_1,
\]
every trajectory on $\Gamma$ is periodic.  We will refer to this case as the \textit{generic case}.  

Hence, we have the following corollary.  
\begin{corollary} Consider the generic case
$\sqrt{K_1}+\sqrt{K_2}<I_1$, keeping $E\equiv 0$.  Then the pairs $(s_1,d_1)$ and $(s_2,d_2)$ execute a \emph{bounded periodic exchange} that does not decay in time.  Both pairs share the same period.  
\end{corollary}

A complete analysis of the periodic dynamics on $\Gamma$, including the reduction to elliptic integrals, is given in \autoref{appendix:periodic}.

For the asymptotic analysis in $t$ of the dynamics of the two-peakon solution, we will use an equivalent approach to studying the $E\equiv 0$ case in the original variables $m_j, n_j$. Again, dropping the $E$ terms in \eqref{eq:DS}, we obtain
\begin{subequations} \label{eq:DSbis}
\begin{align} 
%\dot x_1&=\frac14\big[(m_1+n_1)+(m_2+n_2) e^{-2\abs{x_1-x_2}}\big], \label{eq:x1}\\
%\dot x_2&=\frac14\big[(m_1+n_1)e^{-2\abs{x_1-x_2}}+(m_2+n_2) \big], \label{eq:x2}\\
\dot m_1&=\frac12 m_1\big[ n_2-m_2 \big], \label{eq:m1bis}\\
\dot n_1&=\frac12 n_1\big[ m_2-n_2\big], \label{eq:n1bis}\\
\dot m_2&=\frac12 m_2\big[ m_1-n_1\big], \label{eq:m2bis} \\
\dot n_2&=\frac12 n_2\big[ n_1-m_1\big]. \label{eq:n2bis}
 \end{align} 
\end{subequations} 
We note that $I_3=m_1m_2n_1n_2$ is now a constant of motion, in addition to the previously mentioned $I_1, K_1, K_2$.  In fact, as shown explicitly in \autoref{appendix:relations},   $I_3=\frac{K_1K_2}{16}$.  
\begin{lemma} Suppose the quadruple $(m_1, n_1, m_2, n_2)$ satisfies the equations \eqref{eq:m1bis}-\eqref{eq:n2bis}.  
If we define $$\hat v_+(t)=\frac{1}{2} (n_1(t)+n_2(t)-m_1(t)-m_2(t)), \quad m_1(t)n_2(t)=\sqrt{I_3} e^{\hat y(t)}, \, \text{ and } m_2(t)n_1(t)=\sqrt{I_3} e^{-\hat y(t)},  $$
where $I_3=m_1m_2n_1n_2$, 
then the pair $(\hat v_+, \hat y)$ satisfies the same Hamiltonian system \eqref{eq:sys} as the pair $(v_+, y)$.  
\end{lemma} 
\begin{proof} 
By straightforward differentiation of $m_1(t)n_2(t)=\sqrt{I_3} e^{\hat y(t)}$ and use of \eqref{eq:DSbis}, we obtain $\hat v_+=\dot{\hat y}$.  
Likewise, 
\[
\dot{\hat {v}}_+=\frac12[m_2n_1-m_1n_2]=-\sqrt{I_3} \sinh \hat y. 
\]
\end{proof} 
Our goal is to identify $\hat v_+(t)=v_+(t)$ and $\hat y(t) =y(t)$.  For this to hold, the pair $(\hat y(t), \hat v_+(t))$ must lie on the same energy level $\cal{E}$ as the pair $(y(t), v_+(t))$.  Moreover, as 
\eqref{eq:ydot2} shows, $y(t)$ is then determined up to a shift.  
\begin{lemma} \label{lem:v+}
Fix $I_3$ and let $(y(t), v_+(t))$ satisfy the Hamiltonian system 
\[ \dot y=v_+, \quad \dot v_+=-\sqrt{I_3} \sinh y, \]
with energy $\cal{E}$.  Choose the constants  of motion $I_1, K_1, K_2$ for the system \eqref{eq:DSbis} so that 
\[ 
I_3=\frac{K_1K_2}{16}, \quad \mathcal{E}=\frac{I_1^2-K_1-K_2}{8}, 
\]
and $\hat y(0)=y(0)$.  Then 
\[ (\hat v_+(t), \hat y (t))=(v_+(t), y(t)), \]

\end{lemma} 
\begin{proposition} \label{prop:m1n1m2n2-d}
\begin{align}
&m_1(t)=\frac{\frac12 I_1-v_+(t)}{1+(G^\infty)^{-1}e^{-y(t)}}, &n_1(t)=\frac{\frac12 I_1+v_+(t)}{1+(G^\infty)^{-1}e^{y(t)}}\\
&m_2(t)=\frac{\frac12 I_1-v_+(t)}{1+G^\infty e^{y(t)}}, &n_2(t)=\frac{\frac12 I_1+v_+(t)}{1+G^\infty e^{-y(t)}}, 
\end{align} 
where $G^\infty=\sqrt{\frac{K_1}{K_2}}$.  
\end{proposition} 
\begin{proof} 
From the equations \eqref{eq:m1bis}-\eqref{eq:n2bis} we obtain 
\[ \frac{d \ln \frac{m_1}{m_2}}{dt}=\hat v_+ \overbrace{=}^{\autoref{lem:v+}}v_+=\dot y. \]
Hence 
\[ \frac{d \ln \frac{m_1}{m_2}}{dt}=\dot y , \]
or, equivalently, 
\[  \frac{d \ln \frac{m_1e^{- y}}{m_2}}{dt}=0,  \]
which yields $\frac{m_1(t)}{m_2(t)}=Ce^{y(t)}$, where $C=\frac{m_1(0)e^{- y(0)}}{m_2(0)}$.    However, taking the ratio of $m_1n_2=\sqrt{I_3} e^{y}$ and $4m_2n_2=K_2$ implies 
that $\frac{m_1e^{- y}}{m_2}=4\frac{\sqrt{I_3}}{K_2}=\sqrt{\frac{K_1}{K_2}}$ (in \autoref{appendix:relations}, we gather useful identities, including the one used in the last step).  
We summarize: 
\[ \frac{m_	1(t)}{m_2(t)}=\sqrt{\frac{K_1}{K_2}} e^{y(t)}. 
\]
Now, we use the relations $\frac{I_1}{2}=\frac{m_1+m_2+n_1+n_2}{2}$ and $v_+=\frac{n_1+n_2-m_1-m_2}{2}$ to obtain 
\[m_1+m_2=\frac{I_1}{2}-v_+, \]
from which the formulas for $m_1, m_2$ follow.  The proof for the pair $(n_1, n_2)$ follows analogous steps.  
\end{proof} 
We recall the linear transformation 
\[
s_j=m_j+n_j, \qquad d_j=m_j-n_j, \qquad j=1,2 
\]
mapping the quadruple $(m_1, n_1, m_2, n_2)$ to the quadruple $(s_1, d_1, s_2, d_2)$.  
To study the asymptotics of the full solution, we will use the explicit formulas from \autoref{prop:m1n1m2n2-d}.  
\begin{definition} \label{def:Y0}
Let $(m_1, n_1, m_2, n_2)$ be given by formulas from \autoref{prop:m1n1m2n2-d} and set 
\begin{equation} 
\bm{Y}_0(t)=(s_1(t), d_1(t), s_2(t), d_2(t)).  
\end{equation} 
\end{definition} 
This will be our preferred parametrization of the curve $\Gamma$.  
\section{Generic case} \label{sec:GC} 
In this section, we study the generic case $\sqrt{K_1^\infty}+\sqrt{K_2^\infty}<I_1$.  
Recall the two--site dynamics in variables $\bm Y=(s_1,d_1,s_2,d_2)\in\R^4$:
\begin{align}
\dot s_1 &= -\tfrac12 s_2 E\,s_1 - \tfrac12 d_1 d_2, &
\dot d_1 &= -\tfrac12 s_2 E\,d_1 - \tfrac12 d_2 s_1, \label{eq:sys1}\\
\dot s_2 &= +\tfrac12 s_1 E\,s_2 + \tfrac12 d_1 d_2, &
\dot d_2 &= +\tfrac12 s_1 E\,d_2 + \tfrac12 d_1 s_2, \label{eq:sys2}
\end{align}
with $I_1=s_1+s_2$ conserved and $E(t)=e^{-2D(t)}\in(0,1]$ integrable in time (\autoref{thm:gap}). Write
\begin{equation}\label{eq:split}
\dot{\bm Y}=F_0(\bm Y)+E(t)\,H(\bm Y),\qquad
\end{equation}
where
\begin{align}
F_0(s_1,d_1,s_2,d_2)&=\big(-\tfrac12 d_1 d_2,\;-\tfrac12 d_2 s_1,\;+\tfrac12 d_1 d_2,\;+\tfrac12 d_1 s_2\big), \label{eq:F0H}
\quad \\
H(s_1,d_1,s_2,d_2)&=\big(-\tfrac12 s_2 s_1,\;-\tfrac12 s_2 d_1,\;+\tfrac12 s_1 s_2,\;+\tfrac12 s_1 d_2\big). \label{eq:Hsys}
\end{align}

Let $K_j(t)=s_j^2-d_j^2=4m_j n_j$ and $K_j^\infty=\lim_{t\to\infty}K_j(t)\in(0,\infty)$ (established earlier). The \emph{limit system} $\dot{\bm Y}=F_0(\bm Y)$ has the invariants $I_1$, $K_1$, $K_2$; thus for fixed $(I_1,K_1^\infty,K_2^\infty)$, such that $\sqrt{K_1^\infty}+\sqrt{K_2^\infty}<I_1$,  it admits a unique nontrivial periodic orbit $\bm Y_0(t)$ of period $T>0$, given by \autoref{def:Y0}, lying on the closed curve  $\Gamma^\infty$ defined in \eqref{eq:Gamma-curve} and studied in \autoref{subsec:Enull}, but now taken for fixed values $K_1=K_1^\infty, K_2=K_2^\infty.$
%\begin{equation}\label{eq:invariants}
%\Gamma=\{\bm Y: s_1+s_2=I_1,\qquad s_1^2-d_1^2=K_1^\infty>0,\qquad s_2^2-d_2^2=K_2^\infty>0, s_1>0, s_2>0\}.  
%\end{equation}
%Let us denote by $\Gamma^\infty\subset \Gamma$ the image of $\bm Y_0(t)$.  
Our goal is to justify the tracking estimate
\begin{equation}
\boxed{\quad
\text{dist}(\bm Y(t), \Gamma^\infty)\;\le\; C \int_t^\infty E(s)\,ds\;\xrightarrow[t\to\infty]{}\;0,
\quad}\label{eq:goal}
\end{equation}
for a suitable $C>0$.

\subsection{Dynamics in terms of \texorpdfstring{$v_+$ and $y$} {v+ and y} }
Recall the peakon dynamical system \eqref{eq:DS} in the original variables $(x_1, x_2, m_1, n_1, m_2, n_2)$: 

\begin{align*} 
\dot x_1&=\frac14\big[(m_1+n_1)+(m_2+n_2) e^{-2\abs{x_1-x_2}}\big], \quad 
&\dot x_2&=\frac14\big[(m_1+n_1)e^{-2\abs{x_1-x_2}}+(m_2+n_2) \big], \\
\dot m_1&=\frac12 m_1\big[ -(m_2+n_2)e^{-2\abs{x_1-x_2}}+(n_2-m_2) \big], \quad
&\dot n_1&=\frac12 n_1\big[ -(m_2+n_2)e^{-2\abs{x_1-x_2}}-(n_2-m_2) \big], \\
\dot m_2&=\frac12 m_2\big[ (m_1+n_1)e^{-2\abs{x_1-x_2}}-(n_1-m_1) \big],   \quad 
&\dot n_2&=\frac12 n_2\big[ (m_1+n_1)e^{-2\abs{x_1-x_2}}+(n_1-m_1)\big],
\end{align*} 
for variables $x_1,x_2, m_1, m_2, n_1, n_2$ subject to the conditions to conditions $m_j(0)>0, n_j(0)>0$ and the ordering condition $x_1(0)<x_2(0)$, 
and the definition of the asymptotic coupling $v_+=\frac12\big ( n_1+n_2-m_1-m_2)$.  

The periodic variables $(v_+(t),y(t))$ together with the slowly varying parameter $G(t)$ provide an explicit asymptotic parametrization of the full two-peakon dynamics. 

\begin{proposition} [Masses in terms of $v_+$ and $y$] \label{prop:m1n1m2n2-c}
\begin{align}
&m_1(t)=\frac{\frac12 I_1-v_+(t)}{1+G^{-1}(t) e^{-y(t)}}, &n_1(t)=\frac{\frac12 I_1+v_+(t)}{1+G^{-1}(t)e^{y(t)}}\\
&m_2(t)=\frac{\frac12 I_1-v_+(t)}{1+G(t) e^{y(t)}}, &n_2(t)=\frac{\frac12 I_1+v_+(t)}{1+G(t)e^{-y(t)}}, 
\end{align} 
where 
\begin{equation} \label{eq:GtoK1K2}
G(t)=\sqrt{\frac{K_1^\infty}{K_2^\infty} }\; e^{\frac12I_1\int_t^\infty E(\tau)\, d\tau}=\sqrt{\frac{K_1(t)}{K_2(t)}}.   
\end{equation}
\end{proposition} 
\begin{proof} 
From the system of ODEs, we obtain
\begin{equation*} 
\frac{d}{dt} \ln \frac{m_1}{m_2}=v_+-\frac 12 I_1 E.
\end{equation*}
Since $\dot y=v_+$, this gives 
\begin{equation}\label{eq:m1/m2}
\frac{d}{dt}  \ln (\frac{m_1}{m_2} e^{-y})=-\frac12 I_1 E. 
\end{equation} 
We intend to integrate from $t$ to $\infty$ which requires that $\frac{m_1(t)e^{-y(t)}}{m_2(t)}$ converge.  
We claim that $\frac{m_1(t)e^{-y(t)}}{m_2(t)}$ converges to $\sqrt{\frac{K_1^\infty}{K_2^\infty}}$ as $t\to \infty$.  Indeed, 
$m_1(t)n_2(t)(1-E(t))e^{-y(t)}=\sqrt{I_3}$, $4m_2(t)n_2(t):=K_2(t)$ converges to $K_2^\infty$, $E(t)\to 0$, therefore
\[ 
\frac{m_1(t) e^{-y(t)}}{m_2(t)} \to \frac{\sqrt{I_3}}{K_2^\infty}\overbrace{=}^{\eqref{eq:I3E}}\sqrt{\frac{K_1^\infty}{K_2^\infty}}.  
\]
Integrating now \eqref{eq:m1/m2} from $t$ to $\infty$ yields 
\[
\frac{m_1(t)}{m_2(t)}=\sqrt{\frac{K_1^\infty}{K_2^\infty}} e^{\frac12 I_1 \int _t^\infty E(\tau)\, d\tau }\, e^{y(t)}:=G(t) e^{y(t)}. 
\]

  Since $I_1=m_1+m_2+n_1+n_2$ and $v_+=\frac12(n_1+n_2-m_1-m_2)$, we obtain
$$ 
\frac{I_1}{2}-v_+=m_1+m_2=m_2\big[1+Ge^y]
$$ 
from which 
$$ 
m_2(t)=\frac{\frac{I_1}{2}-v_+(t)}{1+G(t) e^{y(t)}}
$$ 
follows.  The remaining statements are proved in an analogous way.  For example, the formula for $n_2$ reads
$$ 
n_2(t)=\frac{\frac{I_1}{2}+v_+(t)}{1+G(t) e^{-y(t)}}, 
$$
where we used that $\frac{n_1(t)e^{y(t)}}{n_2(t)} \to \sqrt{\frac{K_1^\infty}{K_2^\infty}}$.  Finally, to show that $G(t)=\sqrt{\frac{K_1(t)}{K_2(t)}}$, use $K_1(t)=K_1(0) e^{-\int_0^t s_2(\tau)E(\tau)\, d\tau}$ and $K_2(t)=K_2(0) e^{\int_0^t s_1(\tau)\, E(\tau)\, d\tau}$ (see, e.g., the proof of \autoref{thm:gap} for an explicit formula for $m_2(t)n_2(t)$). 
\end{proof} 
\subsection{ Asymptotic behaviour of masses} 
In analogy to the periodic field $\bm Y_0(t)$ defined earlier, (see \autoref{def:Y0} ), we now explicitly define 
$\bm Y(t)=(s_1(t), d_1(t), s_2(t), d_2(t))$ as linear combinations of the entries of $(m_1(t), n_1(t), m_2(t), n_2(t))$ given in \autoref{prop:m1n1m2n2-c}.  
The explicit formulas from \autoref{prop:m1n1m2n2-c} allow us to compare the full dynamics directly with the periodic orbit $\Gamma^\infty$. The asymptotic discrepancy is governed entirely by the quantity
\[
G(t)-G^\infty,
\]
which decays because $E(t)\in L^1(0,\infty)$.

\begin{proposition} [Asymptotic behaviour of masses] \label{prop:Gamma-approx}
\begin{equation} \label{eq:Yass}
\text{dist}(\bm Y(t), \Gamma^\infty)\;\le\; C \int_t^\infty E(s)\,ds\;\xrightarrow[t\to\infty]{}\;0,
\end{equation}
for a suitable $C>0$. 
\end{proposition} 
\begin{proof}
Recall that the periodic orbit $\Gamma^\infty$ is parametrized by the limiting formulas obtained from \autoref{prop:m1n1m2n2-c} after replacing $G(t)$ by its limit
\[
G^\infty :=\sqrt{\frac{K_1^\infty}{K_2^\infty}}.
\]
Thus we define
\begin{align*} 
&m_1^\Gamma(t)=\frac{\frac12 I_1-v_+(t)}{1+(G^\infty)^{-1}e^{-y(t)}},  &&n_1^\Gamma(t)=\frac{\frac12 I_1+v_+(t)}{1+(G^\infty)^{-1}e^{y(t)}},\\
&m_2^\Gamma(t)=\frac{\frac12 I_1-v_+(t)}{1+G^\infty e^{y(t)}},  &&n_2^\Gamma(t)=\frac{\frac12 I_1+v_+(t)}{1+G^\infty e^{-y(t)}}.
\end{align*} 
The corresponding quadruple
\[
\bm Y_0(t)=(s_1^\Gamma(t),d_1^\Gamma(t),s_2^\Gamma(t),d_2^\Gamma(t))
\]
lies on the periodic orbit $\Gamma^\infty$.

Next, by \autoref{prop:m1n1m2n2-c},
\[
G(t)=G^\infty \exp\left(\frac12 I_1\int_t^\infty E(\tau)\,d\tau\right).
\]
Hence
\[
G(t)-G^\infty=G^\infty\left[\exp\left(\frac12 I_1\int_t^\infty E(\tau)\,d\tau\right)-1\right].
\]
Since
$\int_t^\infty E(\tau)\,d\tau \to 0$, Taylor expansion gives
\[
G(t)-G^\infty=O\left(\int_t^\infty E(\tau)\,d\tau\right).
\]

We now compare the full masses with their limiting periodic counterparts. For example,
\[
m_1(t)=\frac{\frac12 I_1-v_+(t)}{1+G(t)^{-1}e^{-y(t)}}.
\]
Since $v_+(t)$ and $y(t)$ remain bounded uniformly in time, the function
\[
F(G,y,v)=\frac{\frac12 I_1-v_+}{1+G^{-1}e^{-y}}
\]
is smooth in a neighbourhood of $(G^\infty,y(t),v_+(t))$. 
Therefore,
\[
m_1(t)-m_1^\Gamma(t)=\partial_G F(\widetilde G(t),y(t),v_+(t))\bigl(G(t)-G^\infty\bigr)
\]
for some intermediate value $\widetilde G(t)$ between $G(t)$ and $G^\infty$.

Since $G(t)\to G^\infty>0$ and $(y(t),v_+(t))$ remain bounded, the derivative $\partial_G F$ is uniformly bounded. Consequently,
\[
m_1(t)-m_1^\Gamma(t)=O\left(\int_t^\infty E(\tau)\,d\tau\right).
\]

The same argument applies to $m_2,n_1,n_2$. Since the transformation
\[
(s_j,d_j)=(m_j+n_j,m_j-n_j)
\]
is linear, we obtain
\[
\| \bm Y(t)-\bm Y_0(t)\|
=
O\left(
\int_t^\infty E(\tau)\,d\tau
\right).
\]

Finally, because $\bm Y_0(t)\in\Gamma^\infty$ for all $t$,
\[
\operatorname{dist}(\bm Y(t),\Gamma^\infty)
\le\| \bm Y(t)-\bm Y_0(t)\|,
\]
which proves
\[
\operatorname{dist}(\bm Y(t),\Gamma^\infty)\le C\int_t^\infty E(\tau)\,d\tau \to 0.
\]

\end{proof}

\section{Quantitative separation and trajectory analysis} \label{sec:Dynamics of separation}
\subsection{Relation Between the Distance \texorpdfstring{$D$}{D} and the Accumulated Imbalance of Amplitudes} 
We are interested in relating the distance $D(t)$ to other dynamical quantities to refine our understanding of the asymptotics of peakons.  
The following identity plays a central role in that task.  

\begin{theorem}[Structural identity]
\label{thm:SI}
Define
\[
R(t):=s_2(t)-s_1(t).
\]
Then
\begin{equation}
\label{eq:SI}
e^{2D(t)}-1=\left(e^{2D(0)}-1\right)\exp\left(\frac12\int_0^t R(\tau)\,d\tau\right).
\end{equation}
More generally, for any $t\ge0$ and $h\in\mathbb R$ such that $t+h\ge0$,
\begin{equation}
\label{eq:SI-h}
e^{2D(t+h)}-1=\left(e^{2D(t)}-1\right)\exp\left(\frac12\int_t^{t+h}R(\tau)\,d\tau\right).
\end{equation}
\end{theorem}

\begin{proof}
Since $E=e^{-2D}$,
\[
\frac{1-E}{E}=e^{2D}-1.
\]
Using
\[
\dot E=-\frac12R(1-E)E,
\]
we obtain
\[
\frac{d}{dt}\log\bigl(e^{2D(t)}-1\bigr)=\frac{d}{dt}\log\frac{1-E(t)}{E(t)}=\frac12R(t).
\]
Integrating over $[0,t]$ gives \eqref{eq:SI}, while integration
over the oriented interval from $t$ to $t+h$ gives
\eqref{eq:SI-h}.
\end{proof}

\begin{remark} 
The above identity gives an exact relation between the growth of the distance $D(t)$ and the accumulated imbalance of the amplitudes $s_1(t)$ and $s_2(t)$ represented by $\int_0^t (s_2(\tau)-s_1(\tau))\, d\tau$.  
\end{remark} 

The next approximation lemma is an example of a comparison technique that ``freezes'' the moving coupling $G(t)$ in the formulas for the masses $m_1, m_2, n_1, n_2$ given by the equations in \autoref{prop:m1n1m2n2-c}.  The resulting expressions are then used as periodic comparison functions.  To do so, all expressions must be viewed as functions of $(v_+, y, G)$.  For example, we will write $s_j(v_+, y, G)$ for the frozen value $G$, with the exception of $s_j(v_+(t), y(t), G(t))$ and $R\big(v_+(t),y(t), G(t)\big)$, which we will 
write simply as $s_j(t), R(t)$ to avoid cluttered notation.

We begin with a simple comparison lemma.

\begin{lemma} \label{lem:comparison1} 
Let $\,  0\le a<b<\infty$. Then
\begin{equation}
\label{eq:FrozenR}
\int_a^bR\bigl(v_+(\tau),y(\tau),G(a)\bigr)\,d\tau \le\int_a^b R(\tau)\,d\tau\le\int_a^bR\bigl(v_+(\tau),y(\tau),G(b)\bigr)\,d\tau.
\end{equation}
More generally, if
\[
0\le a'\le a<b\le b'\le\infty,
\]
then
\begin{equation}
\label{eq:FrozenRE}
\int_a^bR\bigl(v_+(\tau),y(\tau),G(a')\bigr)\,d\tau \le\int_a^b R(\tau)\,d\tau\le\int_a^bR\bigl(v_+(\tau),y(\tau),G(b')\bigr)\,d\tau,
\end{equation}
where $G(\infty):=G^\infty$.
\end{lemma}

\begin{proof}
For fixed $v_+$ and $y$, the function
$G\longmapsto s_1(v_+,y,G)$ is increasing, whereas $G\longmapsto s_2(v_+,y,G)$ is decreasing. Consequently,
\[
G\longmapsto
R(v_+,y,G)=s_2(v_+,y,G)-s_1(v_+,y,G)
\]
is decreasing.

Since $G(t)$ is decreasing, for every $\tau\in[a,b]$ we have
\[
G(a)\ge G(\tau)\ge G(b).
\]
It follows that
\[
R\bigl(v_+(\tau),y(\tau),G(a)\bigr) \le R\bigl(v_+(\tau),y(\tau),G(\tau)\bigr) \le R\bigl(v_+(\tau),y(\tau),G(b)\bigr).
\]
Since
\[
R\bigl(v_+(\tau),y(\tau),G(\tau)\bigr)=R(\tau),
\]
integration over $[a,b]$ proves \eqref{eq:FrozenR}.

More generally, if
\[
0\le a'\le a<b\le b'\le\infty,
\]
then
\[
G(a')\ge G(a)\ge G(\tau)\ge G(b)\ge G(b').
\]
The monotonicity of $R$ with respect to $G$ therefore gives
\[
R\bigl(v_+(\tau),y(\tau),G(a')\bigr) \le R(\tau) \le R\bigl(v_+(\tau),y(\tau),G(b')\bigr).
\]
Integrating over $[a,b]$ proves \eqref{eq:FrozenRE}.
\end{proof}

The following formula can be readily inferred from \eqref{eq:S1S2} in \autoref{appendix:positions}: 
\begin{equation} \label{eq:Rab-exact}
\int_a^bR\big(v_+(\tau), y(\tau), G\big)=c(G) (b-a)+2P_{\mathrm{per}}(b;G)-2P_{\mathrm{per}}(a;G), 
\end{equation} 
where 
\[ c(G):=A(G) \frac{\Pi(n(G)\mid m)}{K(m)}, \quad A(G):=I_1\frac{1-G}{1+G}, \quad  n(G):=-\frac{4k^2G}{(1+G)^2}, \]
$K(m)$ and $\Pi(n\mid m)$ are complete elliptic integrals of the first and third kinds, respectively, and $P_{\mathrm{per}}(t;G)$ is a periodic function with period $T_A=\frac{T}{2}$ such that $P_{\mathrm{per}}(0;G)=0$.  
\eqref{eq:Rab-exact} allows one to strengthen the final statement of the previous lemma by evaluating the right and left hand sides.  

\begin{theorem}
\label{thm:frozen-integral-comparison}
Let
\[
0\le a'\le a<b\le b'\le\infty.
\]
Then
\begin{equation}\label{eq:frozen-integral-comparison}
c(G(a'))(b-a)+2\Delta_{[a,b]}P_{\mathrm{per}}(\,\cdot\,;G(a'))\le\int_a^bR(\tau)\,d\tau \le c(G(b'))(b-a)+2\Delta_{[a,b]}P_{\mathrm{per}}(\,\cdot\,;G(b')),
\end{equation}
where
\[
\Delta_{[a,b]}P_{\mathrm{per}}(\,\cdot\,;G):=P_{\mathrm{per}}(b;G)-P_{\mathrm{per}}(a;G),
\]
and \(G(\infty)=G^\infty\).

In particular, taking \(b=t\) and \(b'=\infty\), we obtain
\begin{equation} c(G(a'))(t-a)+2\Delta_{[a,t]}P_{\mathrm{per}}(\,\cdot\,;G(a'))\le \int_a^tR(\tau)\,d\tau \le c(G^\infty)(t-a)+2\Delta_{[a,t]}P_{\mathrm{per}}(\,\cdot\,;G^\infty).
\label{eq:FrozenRE2}
\end{equation}
\end{theorem}
\begin{proof}
Apply Lemma~\ref{lem:comparison1} and use, for every fixed \(G>0\),
\[
\int_a^b
R\bigl(v_+(\tau),y(\tau),G\bigr)\,d\tau=c(G)(b-a)+2\Delta_{[a,b]}P_{\mathrm{per}}(\,\cdot\,;G).
\]
\end{proof}
\begin{theorem}[Exponential decay of the interaction]
\label{thm:E-estimates}
Assume the generic case $\sqrt{K_1^\infty}+\sqrt{K_2^\infty}<I_1$. Then
\[
G^\infty<1, \qquad\text{equivalently}\qquad K_1^\infty<K_2^\infty.
\]
Moreover, there exist $t_*>0$, $\lambda_*>0$, and constants
$C_E,C_I>0$ such that, for all $t\ge t_*$,
\[
E(t)\le C_Ee^{-\lambda_*t},
\]
and
\[
\int_t^\infty E(\tau)\,d\tau \le C_Ie^{-\lambda_*t}.
\]
In particular,
\[
\int_0^\infty \left(\int_t^\infty E(\tau)\,d\tau \right)dt<\infty.
\]
\end{theorem}

\begin{proof}
Recall 
\[
c(G)=I_1\frac{1-G}{1+G}\, \frac{\Pi(n(G)\mid m)}{K(m)},\qquad  n(G)=-\frac{4k^2G}{(1+G)^2}.
\]
Setting $a=0$ in the frozen-parameter estimate  \eqref{eq:FrozenRE2} yields 

\[
\int_0^tR(\tau)\,d\tau \le c(G^\infty)t+2\Delta_{[0,t]}P_{\mathrm{per}}(\,\cdot \, ;G^\infty).  
\]
The periodic term is bounded. On the other hand, the structural identity
\[
e^{2D(t)}-1=\left(e^{2D(0)}-1\right)\exp\left(\frac12\int_0^tR(\tau)\,d\tau \right)
\]
and $D(t)\to\infty$ imply
\[
\int_0^tR(\tau)\,d\tau\to+\infty, 
\]
hence $c(G^\infty)>0$.  
Since
$\Pi(n(G^\infty)\mid m)>0$ and $K(m)>0$, 
\[I_1 \frac{1-G^\infty}{1+G^\infty}>0, \] 
therefore $G^\infty<1$.

Now choose $t_*>0$ so that $G(t_*)<1$, and put
\[
c_*:=c(G(t_*))>0.
\]
For $t\ge t_*$, choose $a=t_*, b=t$ in the frozen-parameter comparison \eqref{eq:FrozenRE2} on the interval
$[t_*,t]$ to get 
\[
\int_{t_*}^tR(\tau)\,d\tau
\ge
c_*(t-t_*)+2\Delta_{[t_*,t]}P_{\mathrm{per}}(\, \cdot \,; G(t_*)).
\]
Since the last term is bounded, there exists $C_*>0$ such that
\[
\int_0^tR(\tau)\,d\tau \ge c_*t-C_*
\]
for all \(t\ge t_*\).

Using the structural identity again,
\[
E(t)=\frac{1}{1+\beta \exp\left(\frac12\int_0^tR(\tau)\,d\tau\right)},\qquad\beta=e^{2D(0)}-1,
\]
and hence
\[
E(t)
\le\frac1\beta \exp\left(-\frac12\int_0^tR(\tau)\,d\tau \right)\le \frac{e^{C_*/2}}{\beta}e^{-c_*t/2}.
\]
Thus the first estimate holds with
\[
\lambda_*:=\frac{c_*}{2},
\qquad
C_E:=\frac{e^{C_*/2}}{\beta}, 
\]
and the second estimate holds with 
\[C_I:=\frac{C_E}{\lambda_*}. 
\]
Finally, to prove the last estimate, we note that the integral over $[t_*,\infty)$ is finite by the exponential estimate, while the integral over $[0,t_*]$ is finite because $\int_t^\infty E(\tau)\, d\tau$ is finite and bounded there.  
Equivalently, one can apply Tonelli's theorem to the positive function $F(t,\tau)=\mathbf{1}_{t\le \tau} E(\tau)$ to obtain the identity $\int_0^\infty\left(\int_t^\infty\, E(\tau)\, d\tau\right)\, dt=\int_0^\infty \tau E(\tau)\, d\tau$ and 
use the estimate for $E(\tau)$.  
\end{proof}

We can now refine the estimate in 
\autoref{prop:Gamma-approx}, namely, 
\[
\text{dist}(\bm Y(t), \Gamma^\infty)\;\le\; C \int_t^\infty E(s)\,ds\;\xrightarrow[t\to\infty]{}\;0.
\]  

\begin{corollary}
\label{cor:exponential-approach-Gamma}
Under the assumptions of \autoref{thm:E-estimates}, there exists $\lambda_*>0$ such that 
\[
\operatorname{dist}\bigl(\bm Y(t),\Gamma^\infty\bigr)
=
O(e^{-\lambda_*t})
\qquad\text{as }t\to\infty.
\]
\end{corollary}
The rate of change of the separation \(D(t)\) is governed by
\(R(t)\). Indeed, recalling \eqref{eq:dotD},
\[
\dot D(t)=\frac14R(t)(1-E(t)).
\]
We prove below that although $R(t)$ may change sign, its integral over one period $T_A$ has a definite sign once $G$ lies entirely on one
side of $1$. Consequently, the separation is monotone when sampled at time increments of \(T_A\).

\begin{theorem}\label{thm:Tsampling} 
    Let $T$ be the period of $v_+$ and $y$ and set $T_A=\frac{T}{2}$.  
    \begin{enumerate}
        \item The condition $K_2(t_-)<K_1(t_-)$, or equivalently $G(t_-)>1$, at some time $t_-$ implies that, for all prior times $t$, the distance decreases over each period $T_A$, i.e., for $T_A<t<t_-$, $D(t) < D(t-T_A)$. Moreover, there exist $\lambda_->0$ and a constant $B_E>0$ such that, for all $T_A<t<t_-$,
        \[
E(t)\ge \frac{1}{1+B_Ee^{-\lambda_-t}}.
\]
        
        \item The condition $K_2(t_+)>K_1(t_+)$, or equivalently $G(t_+)<1$, at some time $t_+$ implies that, for all subsequent times $t$, the distance increases over each period $T_A$, i.e., for $t>t_+$, $D(t)<D(t+T_A)$. 
        Moreover, there exist $\lambda_+>0$, and a constant 
$C_E>0$ such that, for all $t\ge t_+$,
\[
E(t)\le C_Ee^{-\lambda_+t},
\]
    \end{enumerate}
    \label{D_thm}
\end{theorem}

\begin{proof}
Suppose first that \(G(t_-)>1\). Since \(G\) is decreasing,
\[
G(t)\ge G(t_-)>1 \qquad\text{for }t\le t_-.
\]
In particular,
\[
c(G(t_-))<0.
\]

Let $T_A<t\le t_-$. Applying
\eqref{eq:frozen-integral-comparison} on the interval
$[t-T_A,t]\), with \(b'=t_-$, gives
\[
\int_{t-T_A}^tR(\tau)\,d\tau \le c(G(t_-))T_A+2\Delta_{[t-T_A,t]}P_{\mathrm{per}}(\,\cdot\,;G(t_-)).
\]
The periodic increment vanishes, and therefore
\[
\int_{t-T_A}^tR(\tau)\,d\tau\le c(G(t_-))T_A<0.
\]
By the structural identity \eqref{eq:SI-h},
\[
\begin{aligned}
e^{2D(t)}-1&=\bigl(e^{2D(t-T_A)}-1\bigr)\exp\left(\frac12\int_{t-T_A}^{t}R(\tau)\,d\tau\right)\\
&<e^{2D(t-T_A)}-1.
\end{aligned}
\]
Hence
\[
D(t)<D(t-T_A).
\]

To estimate \(E\), apply the upper comparison on \([0,t]\),
again with \(b'=t_-\):
\[
\int_0^tR(\tau)\,d\tau
\le
c(G(t_-))t
+
2\Delta_{[0,t]}
P_{\mathrm{per}}(\,\cdot\,;G(t_-)).
\]
Let
\[
\lambda_-:=-\frac12c(G(t_-))>0.
\]
Since the periodic increment is bounded, there exists \(M_->0\)
such that
\[
\frac12\int_0^tR(\tau)\,d\tau
\le -\lambda_-t+M_-.
\]
Consequently,
\[
E(t)=\frac{1}{1+\beta\exp\left(\frac12\int_0^tR(\tau)\,d\tau\right)} \ge\frac{1}{1+B_Ee^{-\lambda_-t}},
\]
where
\[
B_E:=\beta e^{M_-}.
\]

Now suppose that $G(t_+)<1$. Then
\[
c(G(t_+))>0.
\]
For $t\ge t_+$, apply the lower comparison on
$[t,t+T_A]$, with $a'=t_+$. This gives
\[
\int_t^{t+T_A}R(\tau)\,d\tau \ge c(G(t_+))T_A+2\Delta_{[t,t+T_A]} P_{\mathrm{per}}(\,\cdot\,;G(t_+)).
\]
Again the periodic increment vanishes, so
\[
\int_t^{t+T_A}R(\tau)\,d\tau
\ge c(G(t_+))T_A>0.
\]
The structural identity \eqref{eq:SI-h} therefore implies
\[
D(t+T_A)>D(t).
\]

Finally, repeating the lower-bound argument in the proof of
\autoref{thm:E-estimates}, with $t_+$ as the freezing time,
gives constants $C_E>0$ and
\[
\lambda_+:=\frac12c(G(t_+))>0
\]
such that
\[
E(t)\le C_Ee^{-\lambda_+t}, \qquad t\ge t_+.
\]
\end{proof}

\begin{remark} \label{rem: Close Encounters}
\autoref{thm:Tsampling} reveals that the sign of $G-1$ determines the long-term trend of the peak separation when sampled over one period $T_A$.
Symbolically, 
\begin{align} 
G>1 \Rightarrow D(t+T_A)<D(t), \\
G<1 \Rightarrow D(t+T_A)>D(t). 
\end{align} 
Thus, $G=1$, or equivalently, $K_1=K_2$, marks the transition between two dynamical regimes: period-by-period contraction and period-by-period expansion of the peak separation.  
From a physical point of view, it is natural to ask whether two separated peakons will have a "close encounter" before their distance grows to infinity. 
For the scalar Camassa-Holm peakons, this question is essentially determined by the ordering of the amplitudes. In the present case of peakons with an internal dynamics, the situation is more subtle, since it is 
governed by a slowly varying quantity $G(t)$. \autoref{thm:Tsampling} provides a means to study this question. If $K_2(0)>K_1(0)$, then the peak separation increases, on average, from the very beginning, making a close encounter impossible. On the other hand, the condition $K_1(0)>K_2(0)$ guarantees an initial time regime of period-by-period contraction. Whether this contraction results in an appreciable reduction of separation depends on how 
long the peakons stay in the region $G>1$. For example, if $K_1(0)>K_2(0)$ yet the difference $|K_2(0)-K_1(0)|$ is sufficiently small, the transition $G=1$ might occur while the peakons are widely separated.  
The peakons do initially approach each other, but the encounter is hardly "close". Thus, $K_1(0)>K_2(0)$ can be viewed as a necessary condition for substantial approach rather than a guarantee of it.  
   
\end{remark}     
\subsection{Asymptotic behaviour of positions}
\label{sec:positions}

\begin{proposition}[Asymptotic behaviour of the peakon trajectories]
\label{prop:x-asymptotics}
Assume
\[
\sqrt{K_1^\infty}+\sqrt{K_2^\infty}<I_1,
\]
and let
\[\bm Y_0(t)=\bigl(s_1^\Gamma(t),d_1^\Gamma(t),s_2^\Gamma t),d_2^\Gamma(t)
\bigr)
\in\Gamma^\infty
\]
denote the periodic solution on the limiting orbit
\(\Gamma^\infty\) corresponding to the periodic variables
$(v_+(t),y(t))$.  The full solution \(\bm Y_0(t)\) has period
$T$, while the amplitudes $s_j^\Gamma(t)$ have period
\[
T_A=\frac{T}{2}.
\]

Then there exist constants $c_1,c_2\in\mathbb R$ and
$T_A$-periodic functions $Q_1,Q_2$  such that
\begin{equation}
x_1(t)=c_1+\frac14\langle s_1^\Gamma\rangle t
+\frac14Q_1(t)+o(1),
\label{eq:x1asym}
\end{equation}
and
\begin{equation}
x_2(t)=c_2+\frac14\langle s_2^\Gamma\rangle t
+\frac14Q_2(t)+o(1),
\label{eq:x2asym}
\end{equation}
as $t\to\infty$, where
\[
\langle s_j^\Gamma\rangle=\frac1{T_A}\int_0^{T_A}s_j^\Gamma(\tau)\,d\tau
\]
is the time average over one period of the amplitudes.

In particular,
\[
x_j(t)=v_j^\infty t+c_j+\text{\rm a bounded periodic function}
+o(1),
\]
where
\begin{align}
v_1^\infty&=\frac14\langle s_1^\Gamma\rangle=
\frac{I_1}{8}
\left(1-\frac{1-G^\infty}{1+G^\infty}\frac{\Pi(n(G^\infty)\mid m)}{K(m)}\right),
\label{eq:v1-infty}
\\
v_2^\infty&=\frac14\langle s_2^\Gamma\rangle=\frac{I_1}{8}\left(1+\frac{1-G^\infty}{1+G^\infty}\frac{\Pi(n(G^\infty)\mid m)}{K(m)}\right).
\label{eq:v2-infty}
\end{align}
Here \(K(m)\) and \(\Pi(n\mid m)\) denote the complete elliptic
integrals of the first and third kinds in the elliptic-parameter
convention, with
\[
m=-k^2,\qquad n(G)=-\frac{4k^2G}{(1+G)^2}.
\]
\end{proposition}

\begin{proof}
Recall that
\[
\dot x_1=\frac14(s_1+s_2E),
\qquad
\dot x_2=\frac14(s_1E+s_2),
\]
where $E(t)=e^{-2D(t)}$.

By \autoref{prop:Gamma-approx},
\[
s_j(t)
=
s_j^\Gamma(t)
+
O\left(
\int_t^\infty E(\tau)\,d\tau
\right).
\]
Since the amplitudes \(s_j(t)\) are bounded, it follows that
\[
\dot x_1(t)
=
\frac14s_1^\Gamma(t)+r_1(t),
\]
where
\[
r_1(t)=O(E(t))
+
O\left(
\int_t^\infty E(\tau)\,d\tau
\right).
\]
Similarly,
\[
\dot x_2(t)
=
\frac14s_2^\Gamma(t)+r_2(t),
\]
where
\[
r_2(t)
=
O(E(t))
+
O\left(
\int_t^\infty E(\tau)\,d\tau
\right).
\]

By \autoref{thm:E-estimates}, both \(E(t)\) and
\[
t\longmapsto\int_t^\infty E(\tau)\,d\tau
\]
are integrable on \([0,\infty)\). Hence
\[
r_j\in L^1(0,\infty).
\]
Define
\[
C_j=\int_0^\infty r_j(\tau)\,d\tau.
\]
Then
\[
\int_0^t r_j(\tau)\,d\tau
=
C_j-\int_t^\infty r_j(\tau)\,d\tau
=
C_j+o(1).
\]
Consequently,
\[
x_j(t)
=
x_j(0)
+\frac14\int_0^t s_j^\Gamma(\tau)\,d\tau
+C_j+o(1).
\]

Since \(s_j^\Gamma\) is \(T_A\)-periodic,
\[
\int_0^t s_j^\Gamma(\tau)\,d\tau
=
\langle s_j^\Gamma\rangle t+Q_j(t),
\]
where
\[
Q_j(t)
=
\int_0^t
\left(
s_j^\Gamma(\tau)-\langle s_j^\Gamma\rangle
\right)
d\tau.
\]
The vanishing mean of the integrand implies
\[
Q_j(t+T_A)=Q_j(t),
\]
so \(Q_j\) is \(T_A\)-periodic.

Setting
\[
c_j=x_j(0)+C_j
\]
gives \eqref{eq:x1asym}--\eqref{eq:x2asym}. The formulas for
\(\langle s_j^\Gamma\rangle\), and hence for the asymptotic
velocities, are derived in \autoref{appendix:positions}; see
\eqref{eq:s1ave} and \eqref{eq:s2ave}.
\end{proof}
\section{Degenerate/Boundary case}\label{sec:BC} 
\begin{lemma}[Boundary condition and reduction to equal masses] \label{lem:BC}
Assume that all masses are positive. Then the following statements
are equivalent:
\begin{enumerate}
\item
\[
\sqrt{K_1^\infty}+\sqrt{K_2^\infty}=I_1.
\]
\item
\[
m_j(t)=n_j(t),\qquad j=1,2,\quad t\geq0.
\]
\item
\[
m_j(0)=n_j(0),\qquad j=1,2.
\]
\end{enumerate}
\end{lemma}

\begin{proof}
Assume first that
\[
\sqrt{K_1^\infty}+\sqrt{K_2^\infty}=I_1.
\]
The invariant relations (see \autoref{appendix:relations})
\[
I_3=\frac{K_1^\infty K_2^\infty}{16},
\qquad
\mathcal E
=
\frac{I_1^2-K_1^\infty-K_2^\infty}{8}
\]
give
\[
\mathcal E=\frac{
\left(\sqrt{K_1^\infty}+\sqrt{K_2^\infty}\right)^2
-K_1^\infty-K_2^\infty
}{8}=\frac{\sqrt{K_1^\infty K_2^\infty}}{4}=\sqrt{I_3}.
\]
On the other hand,
\[
\mathcal E=
\frac12v_+^2+\sqrt{I_3}\cosh y
\geq \sqrt{I_3},
\]
with equality if and only if $v_+=0$ and $y=0$. Hence
\[
v_+(t)\equiv0,\qquad y(t)\equiv0.
\]

Since $y=0$, the definition of $y$ gives
\[
m_1n_2=m_2n_1.
\]
By positivity,
\[
\frac{m_1}{n_1}=\frac{m_2}{n_2}=:q.
\]
Moreover,
\[
0=v_+
=\frac12(n_1+n_2-m_1-m_2)
=\frac12(1-q)(n_1+n_2).
\]
Since $n_1+n_2>0$, we obtain $q=1$, and therefore
\[
m_j(t)=n_j(t),\qquad j=1,2.
\]
Thus $1\Rightarrow2$.

The implication $2\Rightarrow3$  is immediate.

Finally, assume $m_j(0)=n_j(0)$, or equivalently
$d_1(0)=d_2(0)=0$. The variables $d_j=m_j-n_j$ satisfy
\[
\dot d_1
=
-\frac12s_2E\,d_1-\frac12s_1d_2,
\qquad
\dot d_2
=
\frac12s_1E\,d_2+\frac12s_2d_1.
\]
By uniqueness, \(d_1(t)=d_2(t)=0\) for all \(t\geq0\). Hence
\[
s_j(t)=\sqrt{K_j(t)},
\]
and therefore
\[
I_1=s_1(t)+s_2(t)
=\sqrt{K_1(t)}+\sqrt{K_2(t)}.
\]
Passing to the limit \(t\to\infty\) yields
\[
I_1=\sqrt{K_1^\infty}+\sqrt{K_2^\infty}.
\]
Thus \(3\Rightarrow1\).
\end{proof}

By \autoref{lem:BC}, the boundary condition
\[
\sqrt{K_1^\infty}+\sqrt{K_2^\infty}=I_1
\]
is equivalent to
\[
m_j(t)=n_j(t),\qquad j=1,2.
\]
Hence \(v\equiv0\), and the system reduces to the standard positive
two-peakon Camassa--Holm flow. 

The long-time asymptotics of that flow are classical; see \cite{camassa-holm} or the results in \cite{beals-sattinger-szmigielski:Stieltjes} based on the spectral approach. In particular,
\[
K_2^\infty>K_1^\infty,
\qquad
s_j(t)\to \sqrt{K_j^\infty},
\]
and
\[
x_j(t)=\frac14\sqrt{K_j^\infty}\,t+c_j+o(1).
\]
Consequently,
\[
D(t)
=
\frac{\sqrt{K_2^\infty}-\sqrt{K_1^\infty}}{4}\,t
+(c_2-c_1)+o(1).
\]

\section{Gallery of pictures illustrating the dynamics of two peakons} \label{sec:gallery} 
Our gallery consists of 5 cases, with the first two representing how \eqref{eq:2CH} can resemble the CH equation, and the last 3 emphasizing how the sign of $G-1$ directly controls the peak separation $D$.

\subsection{Case 1; Boundary (CH) Case}
 $\mathbf{(x_1,x_2,n_1,m_1,n_2,m_2)(0)=(0, 5, 15, 15, 10, 10),\hspace{0.2cm} k=0}$ \\
In the first case, we choose $m_j(0)=n_j(0)$. By the results of the last section, we expect that $x_j$, $m_j$, and $n_j$ will evolve according to an ordinary CH two-peakon equation, and indeed, the figures show typical CH dynamics for both positions and amplitudes.
\begin{figure}[H]
    \centering
    \begin{subfigure}{0.48\textwidth}
        \includegraphics[width=\linewidth]{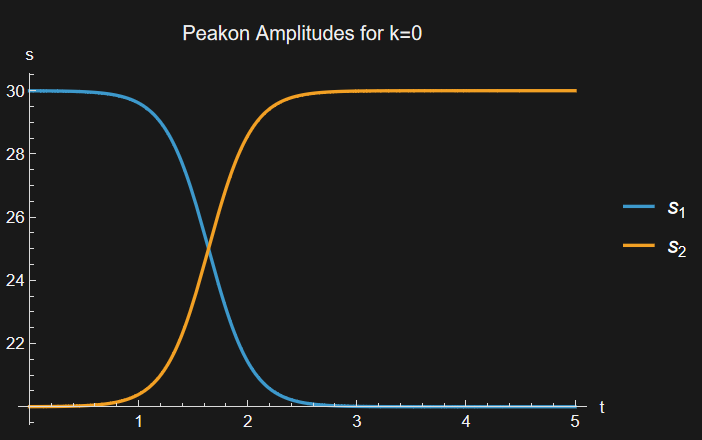}
        \caption{Peakon Amplitudes for $k=0$, Showing Typical CH Behavior}
    \end{subfigure}
    \hfill
    \begin{subfigure}{0.48\textwidth}
        \includegraphics[width=\linewidth]{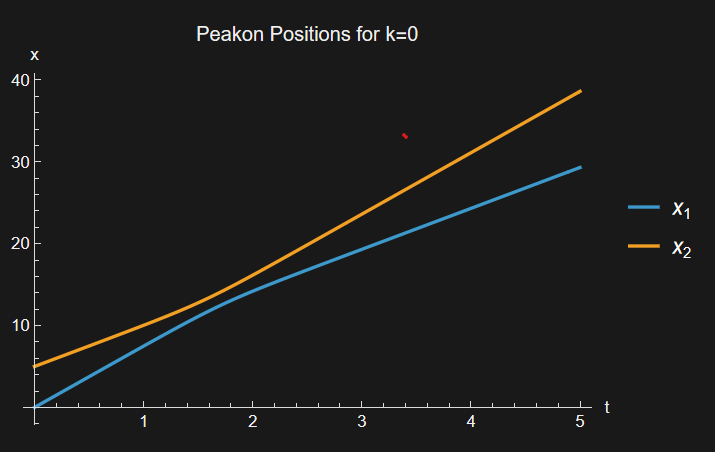}
        \caption{Peakon Positions for $k=0$}
    \end{subfigure}
    
\end{figure}
\subsection{Case 2; Weak Oscillations}
$\mathbf{(x_1,x_2,n_1,m_1,n_2,m_2)(0)=({0, 5, 14, 16, 11, 9}),\hspace{0.2cm} k=0.084}$ \\
The second case represents a slight perturbation of the CH equation, achieved by the choice of small $k=0.084$. This results in small amplitude oscillations via $v_+(t)$ and $y(t)$, though the CH terms $s_jE$ still dominate the dynamics. We note that for early times $t<1$ the masses evolve on an approximately periodic orbit, due to the relatively large initial separation ($5$ units).  Indeed, when the initial separation \(D(0)=x_2(0)-x_1(0)\) is large, the interaction parameter
\[
E(0)=e^{-2D(0)}
\]
is already small. Consequently, the full dynamics start close to the decoupled system $E\equiv 0$. In this regime, the amplitudes initially shadow the periodic orbit
\[
\Gamma_0=\Gamma(I_1,K_1(0),K_2(0)).
\]
The long-time asymptotic orbit studied in Section~\ref{sec:GC} is instead
\[
\Gamma^\infty=\Gamma(I_1,K_1^\infty,K_2^\infty).
\]
Thus, the visible oscillations in the numerical plots can be understood as finite-time manifestations of the same periodic mechanism that governs the asymptotic regime: large spatial separation suppresses the direct interaction term $E$, allowing the internal degrees of freedom to evolve approximately according to the periodic $E\equiv 0$ dynamics.

\begin{figure}[H]
    \centering
    \begin{subfigure}{0.48\textwidth}
        \includegraphics[width=\linewidth]{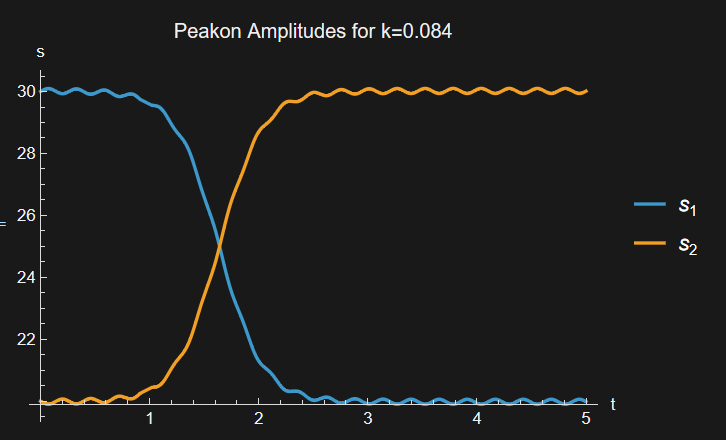}
        \caption{Peakon Amplitudes for $k=0.084$}
    \end{subfigure}
     \begin{subfigure}{0.48\textwidth}
        \includegraphics[width=\linewidth]{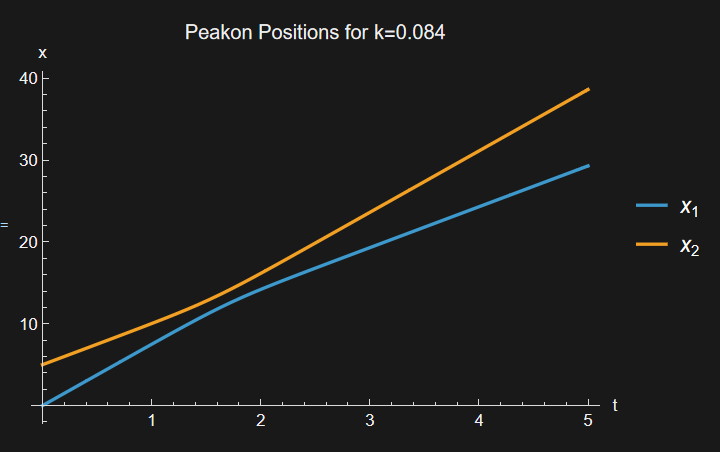}   
        \caption{Peakon Positions for $k=0.084$}
    \end{subfigure}
    \begin{subfigure}{0.48\textwidth}
        \includegraphics[width=\linewidth]{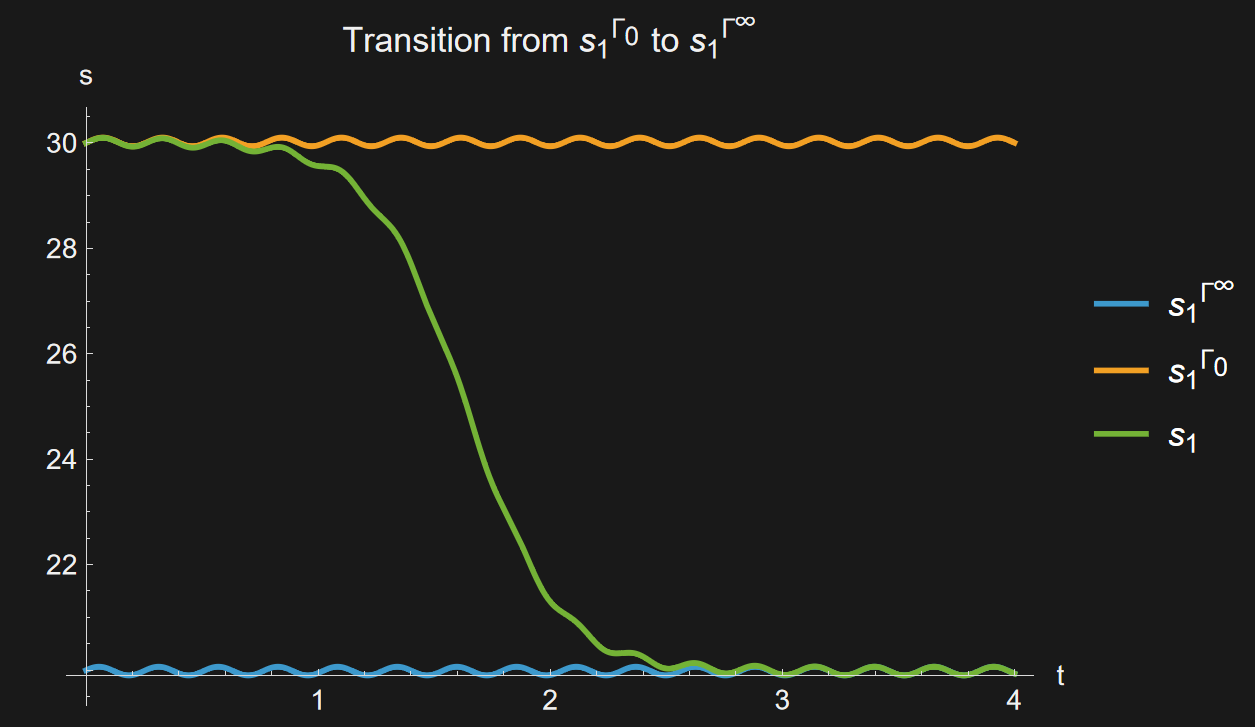}
        \caption{$s_1$ transitions from $s_1^{\Gamma_0}$ to $s_1^{\Gamma^{\infty}}$}
    \end{subfigure}
     \begin{subfigure}{0.48\textwidth}
        \includegraphics[width=\linewidth]{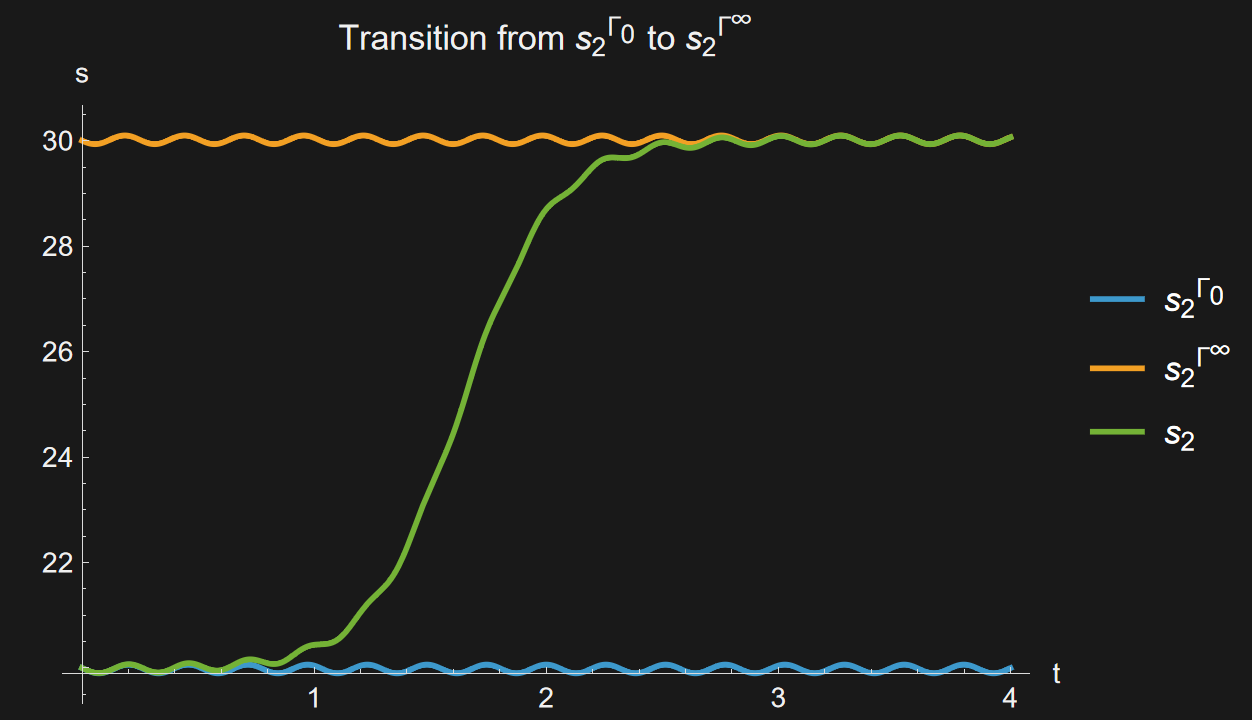} 
        \caption{$s_2$ transitions from $s_2^{\Gamma_0}$ to $s_2^{\Gamma^{\infty}}$}
    \end{subfigure}
\end{figure}

\subsection{Case 3; \texorpdfstring{$G(0)>1$}{G(0)>1} Forces the Peak Gap to Shrink}
$\mathbf{(x_1,x_2,n_1,m_1,n_2,m_2)(0)=({0, 8, 4, 8, 24, 0.1}),\hspace{0.2cm} k=0.808}$ \\
The remaining three cases all differ more radically from CH, as the dynamics are dominated by the oscillations in $v_+$ and $y$, due to the relatively large value of $k$, which is close to $1$ in each case.
One of the striking features of case 3 is that the initial amplitude of the trailing peakon, $s_1(0)=12$, is significantly less than that of the leading peakon, $s_2(0)=24.1$. In the CH case, the peak separation would start to grow to infinity immediately. In contrast, the figures below show that a "close encounter" occurs. Namely, the trailing peakon closes the gap with its neighbor before being deflected, at around $t\in(2.8,3.6)$. By \autoref{thm:Tsampling}, the shrinking in the peak gap $D$ occurs in the region $G>1$. The observed close encounter is therefore a result of the choice $G(0)=3.65>1$. 
\
\begin{figure}[H]
    \centering
    \begin{subfigure}{0.48\textwidth}
        \includegraphics[width=\linewidth]{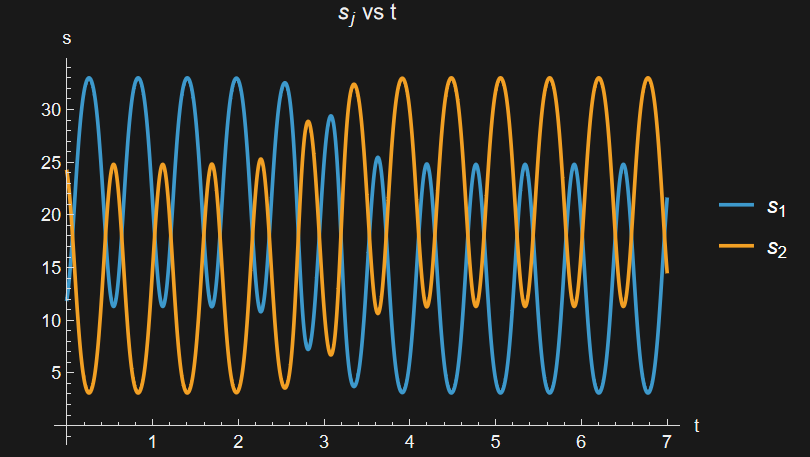}
        \caption{Peakon Amplitudes for $k=0.808$, $G(0)=3.65$}
    \end{subfigure}
     \begin{subfigure}{0.48\textwidth}
        \includegraphics[width=\linewidth]{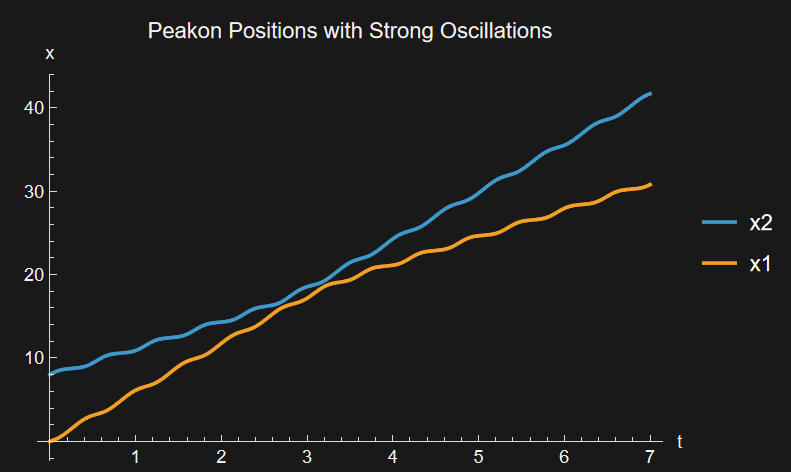}   
        \caption{Peakon Positions for $k=0.808$, $G(0)=3.65$}
    \end{subfigure}
    \label{Case3}
    \end{figure}
\subsection{Case 4; \texorpdfstring{$G(0)<1$}{G(0)<1} Forces an Immediate Increase in the Peak Gap}
$\mathbf{(x_1,x_2,n_1,m_1,n_2,m_2)(0)=({0, 8, 24, 0.1, 4, 8}),\hspace{0.2cm} k=0.808}$\\
The fourth case is the reverse of case 3, where we have chosen $G(0)=0.27<1$, rather than $G(0)>1$. As discussed in \autoref{rem: Close Encounters}, the choice $G(0)<1$ forces the gap $D$ to immediately grow over each period, namely $D(t)<D(t+T_A)$ for all $t$. The growth rate must in particular, satisfy the bound $E\leq C_Ee^{-\lambda_+t}$.
\begin{figure}[h!]
    \centering
    \begin{subfigure}{0.48\textwidth}
        \includegraphics[width=\linewidth]{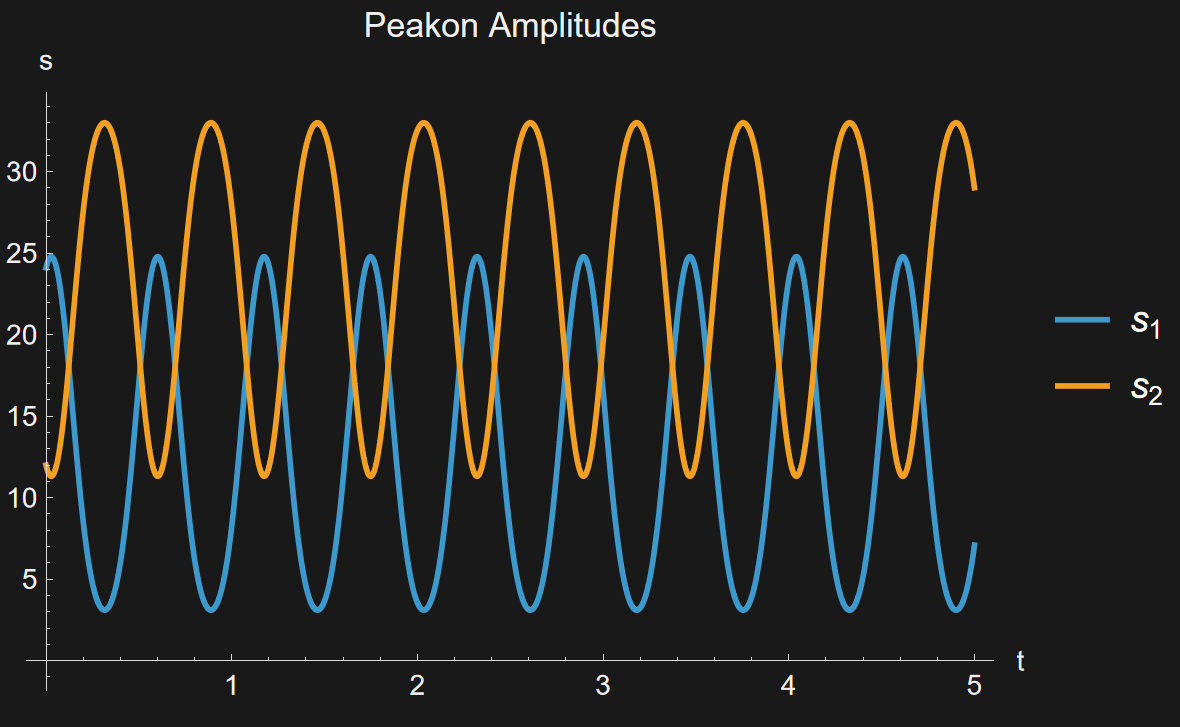}
        \caption{Peakon Amplitudes for $G(0)=0.27$}
    \end{subfigure}
     \begin{subfigure}{0.48\textwidth}
        \includegraphics[width=\linewidth]{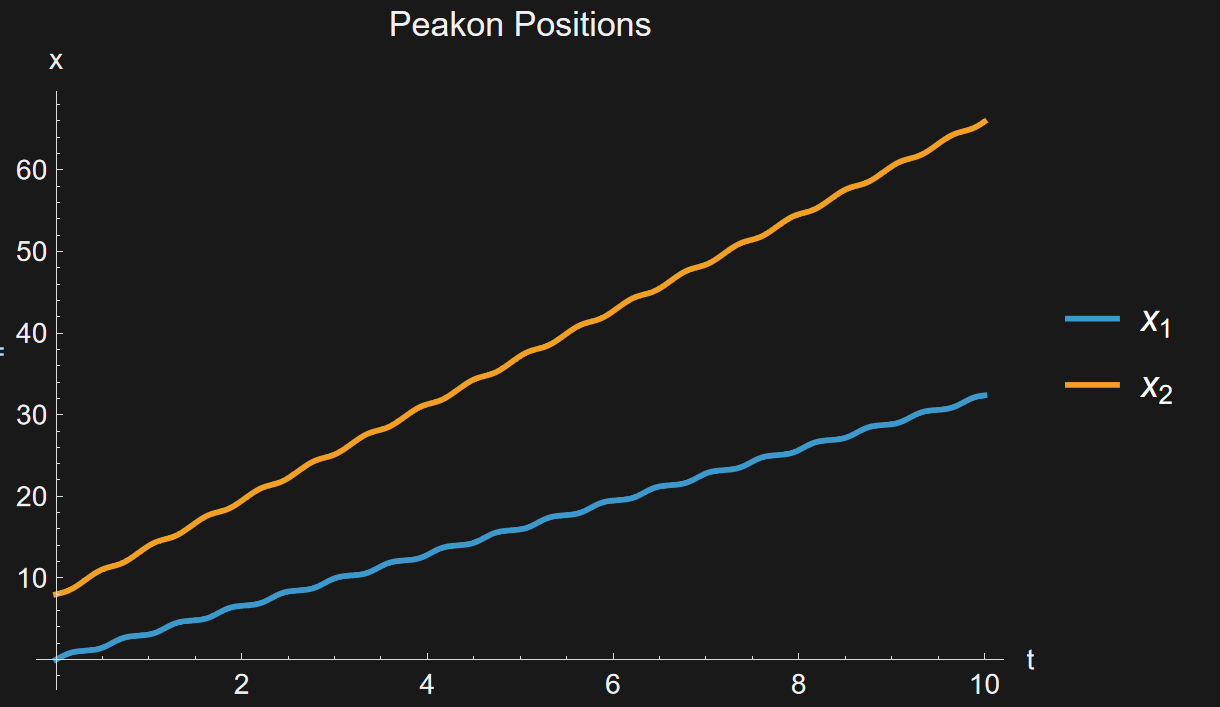}   
        \caption{Peakon Positions for $G(0)=0.27$}
    \end{subfigure}
    \end{figure}
\subsection{Case 5; Peakons Drift with Nearly Equal Average Velocities when 
\texorpdfstring{$G(0)=1$}{G(0)=1}}
$\mathbf{(x_1,x_2,n_1,m_1,n_2,m_2)(0)=({0, 8, 4.5, 2, 1, 9}),\hspace{0.2cm} k=0.943}$\\
The final case is analogous to the CH case in which two separated peaks have roughly equal amplitudes $m_1(0)\approx m_2(0)$, and move together with nearly equal velocities. In the case of \eqref{eq:2CH}, this is achieved by setting $K_1(0)\approx K_2(0)$, or equivalently $G(0)\approx1$. In the following case, we have chosen $G(0)=1$. Because the initial separation $D(0)=8$ is large, the amplitudes start close to the periodic orbit $\Gamma_0$. On this orbit, the average velocities $\frac14\langle s_1^{\Gamma_0} \rangle=\frac14\langle s_2^{\Gamma_0} \rangle$ are identical. Of course, we know that $G$ will immediately decrease, so that the distance begins to grow over each period (\autoref{thm:Tsampling}). However, because $G$ varies slowly in the region with $E<<1$, the amplitudes remain close to the orbit $\Gamma_0$, and the peaks separate exceedingly slowly.
\begin{figure}[H]
    \centering
    \begin{subfigure}{0.48\textwidth}
        \includegraphics[width=\linewidth]{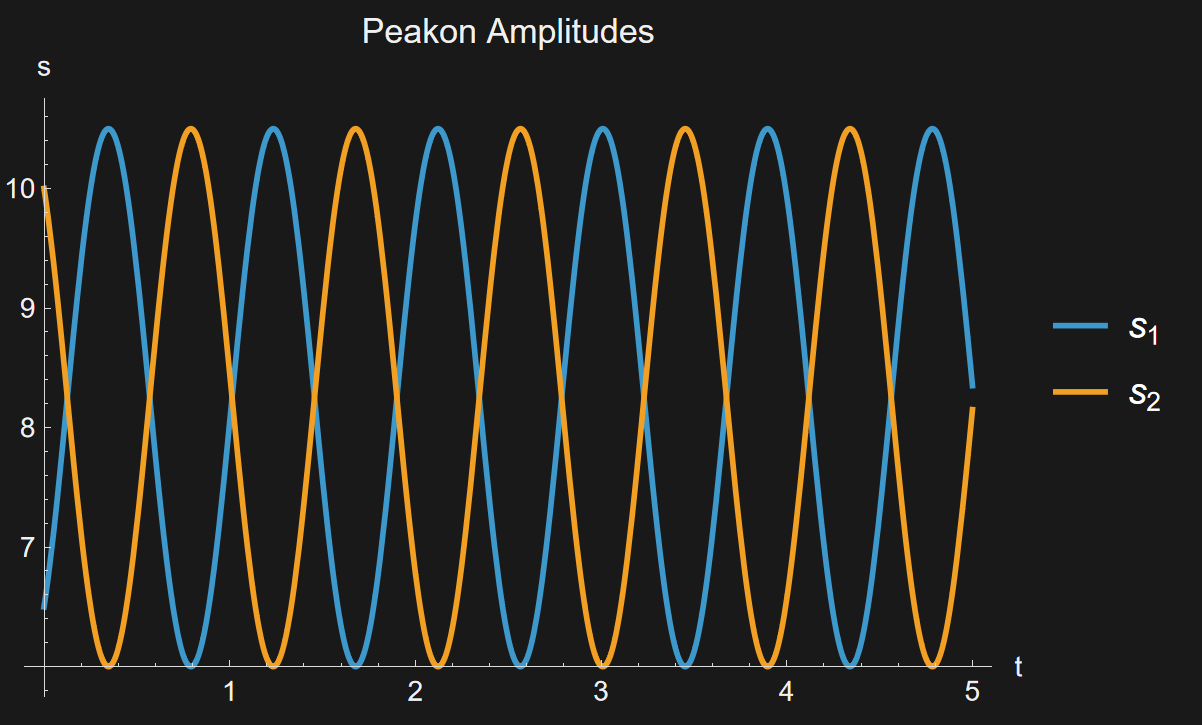}
        \caption{Peakon Amplitudes for $G(0)=1$}
    \end{subfigure}
     \begin{subfigure}{0.48\textwidth}
        \includegraphics[width=\linewidth]{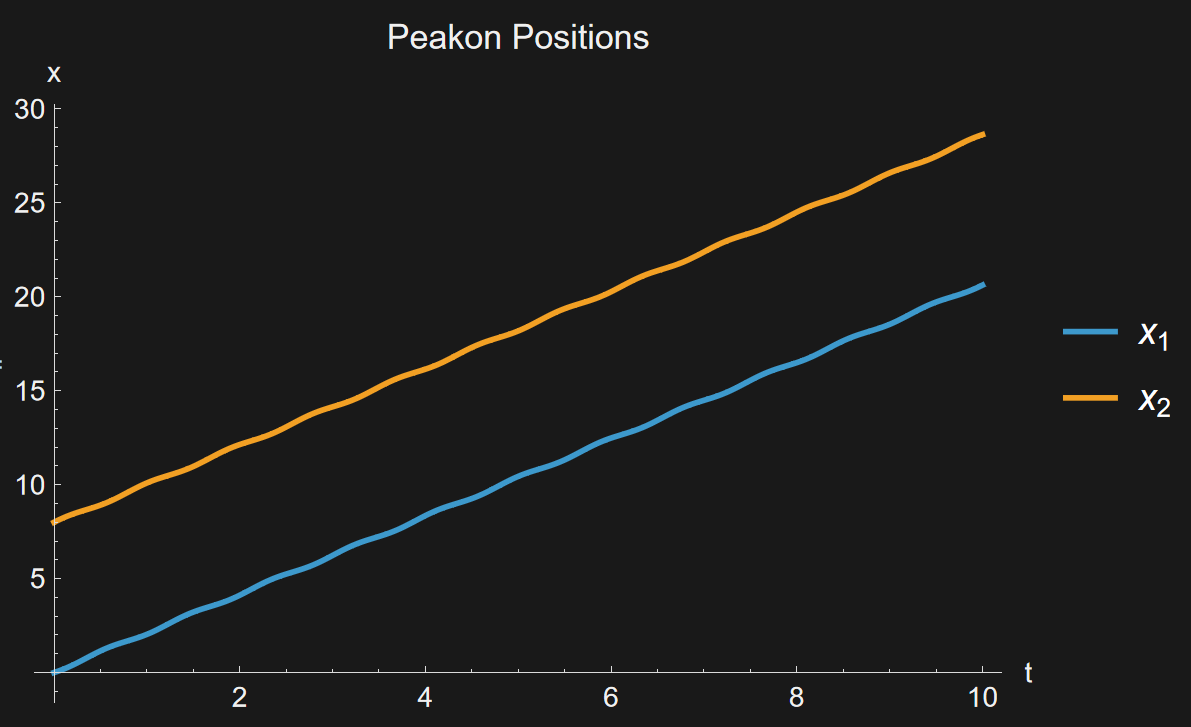}   
        \caption{Peakon Positions for $G(0)=1$}
    \end{subfigure}
    \end{figure}

\section{Conclusions}
 Peakons provide one of the clearest realizations of the particle interpretation of nonlinear waves: the infinite-dimensional PDE admits an invariant finite-dimensional manifold on which the dynamics is governed by a Hamiltonian system for particle positions and momenta. The present study demonstrates that, within the Clifford algebra-based framework developed in \cite{BealsSzm1,BealsSzm2}, the conventional particle description coexists with a robust, long-range coherence of the internal degrees of freedom. The principal result is that these internal degrees of freedom induce a fundamental modification of the long-time particle dynamics of CH peakons. Although the pairwise interactions between peak locations decay exponentially in time, the internal variables evolve toward and remain on a limiting periodic orbit, thereby sustaining a persistent, oscillatory exchange of energy.

There are several directions to pursue after the current work. We list below some of the most pressing issues to investigate:
\begin{itemize}
\item The peakon-antipeakon dynamics: dropping the assumption of positivity of the measure $M$ (see \eqref{eq:M} and \autoref{lem:Lemma1})  might offer a new perspective on the two-component theory.
\item The $N>2$ peakon dynamics, including the derivation of Riemann theta function formulas.
\item The impact of the signature (Minkowski in the present case) on peakon dynamics, not only for $d=2$ but also for general $d$.
\end{itemize}
\section{Acknowledgments}

Jacek Szmigielski's research program is supported by the Natural Sciences and Engineering Research Council of Canada (NSERC). The 2025 and 2026 summer research activities of Alexander Karlson were funded by NSERC through a grant awarded to Jacek Szmigielski.

We wish to express our gratitude to A. Hone and V. Novikov for valuable and stimulating discussions, and for drawing our attention to the reference by Olver and Sokolov \cite{Olver-Sokolov:IS-associative}. We also thank X. Chang for his comments on the manuscript.

\appendix

\section{Periodic dynamics on the asymptotic orbit \texorpdfstring{$\Gamma$}{Gamma}}
\label{appendix:periodic}

In this appendix, we analyze the limiting dynamics obtained by setting
\[
E=e^{-2D}\equiv 0.
\]
The resulting reduced system describes the asymptotic motion on the invariant curve
\[
\Gamma=
\left\{
(s_1,d_1,s_2,d_2):
\ s_1+s_2=I_1,\ 
s_1^2-d_1^2=K_1,\ 
s_2^2-d_2^2=K_2,\ 
s_j>0
\right\}.
\]

We assume throughout the non-degenerate condition
\begin{equation}
\sqrt{K_1}+\sqrt{K_2}<I_1.
\label{eq:nond}
\end{equation}

\subsection{Hyperbolic parametrization}

Introduce hyperbolic coordinates
\begin{equation}
s_j=\sqrt{K_j}\cosh\theta_j,
\qquad
d_j=\sqrt{K_j}\sinh\theta_j,
\qquad j=1,2.
\label{eq:hypcoord}
\end{equation}
Then, the reduced dynamical system becomes
\begin{equation}
\dot\theta_1
=
-\frac12\sqrt{K_2}\sinh\theta_2,
\qquad
\dot\theta_2
=
\frac12\sqrt{K_1}\sinh\theta_1.
\label{eq:thetaflow}
\end{equation}

The conservation law
\[
s_1+s_2=I_1
\]
takes the form
\begin{equation}
\sqrt{K_1}\cosh\theta_1
+
\sqrt{K_2}\cosh\theta_2
=
I_1.
\label{eq:constraint}
\end{equation}

Since $\cosh\ge 1$, condition \eqref{eq:nond} implies that the level set
\eqref{eq:constraint} is a smooth compact closed curve in the
$(\theta_1,\theta_2)$-plane.

\subsection{Reduction to one degree of freedom}

Eliminating $\theta_2$ from \eqref{eq:constraint} gives
\begin{equation}
\cosh\theta_2
=
\frac{I_1-\sqrt{K_1}\cosh\theta_1}{\sqrt{K_2}}.
\label{eq:theta2elim}
\end{equation}

Using
\[
\sinh^2\theta_2=\cosh^2\theta_2-1,
\]
we obtain
\begin{equation}
\sinh^2\theta_2
=
\frac{
(I_1-\sqrt{K_1}\cosh\theta_1)^2-K_2
}{K_2}.
\label{eq:sinh2}
\end{equation}

Substituting into \eqref{eq:thetaflow} yields
\begin{equation}
(\dot\theta_1)^2
=
\frac14
\left(
(I_1-\sqrt{K_1}\cosh\theta_1)^2-K_2
\right).
\label{eq:thetareduced}
\end{equation}

The admissible region is determined by the requirement
\[
\cosh\theta_2\ge 1.
\]
Using \eqref{eq:theta2elim}, this becomes
\[
\cosh\theta_1
\le
\frac{I_1-\sqrt{K_2}}{\sqrt{K_1}}.
\]
Hence the turning points occur at
\begin{equation}
\theta_1
=
\pm a,
\qquad
a
=
\operatorname{arccosh}
\left(
\frac{I_1-\sqrt{K_2}}{\sqrt{K_1}}
\right).
\label{eq:turning}
\end{equation}

The second algebraic root obtained from squaring is inadmissible, since it would force $\cosh\theta_2=-1$.

\subsection{Periodicity}

The vector field \eqref{eq:thetaflow} has an equilibrium only at
\[
(\theta_1,\theta_2)=(0,0).
\]
However, this point belongs to the level set \eqref{eq:constraint} only if
\[
I_1=\sqrt{K_1}+\sqrt{K_2}.
\]
Under the strict inequality \eqref{eq:nond}, the vector field therefore has no zeros on the compact invariant curve \eqref{eq:constraint}.

Consequently, every solution on $\Gamma$ is periodic.

Equivalently, the pairs $(s_1,d_1)$ and $(s_2,d_2)$ execute a bounded periodic exchange with common period.

\subsection{Period integral}

From \eqref{eq:thetareduced}, one quarter-period is
\[
\frac{T}{4}
=
\int_0^a
\frac{
2\,d\theta
}{
\sqrt{
(I_1-\sqrt{K_1}\cosh\theta)^2-K_2
}
}.
\]
Therefore,
\begin{equation}
T=8\int_0^a\frac{d\theta}{\sqrt{(I_1-\sqrt{K_1}\cosh\theta)^2-K_2}},
\label{eq:periodbis}
\end{equation}
where
\[
a=\operatorname{arcosh}\left(\frac{I_1-\sqrt{K_2}}{\sqrt{K_1}}\right).
\]

The integral \eqref{eq:period} is an elliptic integral of the first
kind. Its explicit evaluation is equivalent to the period formula
obtained from the \((v_+,y)\)-reduction in
\autoref{sec:hyperpendulum}.
We show below that
the hyperbolic parametrization is directly related to the variables
introduced in \autoref{sec:hyperpendulum}.  Indeed, 
\[
m_j=\frac{\sqrt{K_j}}{2}e^{\theta_j},
\qquad
n_j=\frac{\sqrt{K_j}}{2}e^{-\theta_j},
\]
and
\[
\sqrt{I_3}=\frac{\sqrt{K_1K_2}}{4},
\]
the definition
\[
m_1n_2=\sqrt{I_3}\,e^y
\]
gives
\[
y=\theta_1-\theta_2.
\]
Moreover,
\[\dot y=\dot\theta_1-\dot\theta_2=-\frac12\left(\sqrt{K_1}\sinh\theta_1+\sqrt{K_2}\sinh\theta_2\right)=v_+.
\]
Thus the $(\theta_1,\theta_2)$-description and the
$(v_+,y)$-description parametrize the same periodic orbit.

Consequently, the period given by \eqref{eq:period} agrees with the
period computed in Section~\ref{sec:hyperpendulum},
\[
T=\frac{4K(-k^2)}{\omega},
\]
where $K$ is written in the elliptic-parameter convention used
throughout the paper.

\subsection{Dynamical description of the curve \texorpdfstring{$\Gamma$}{Gamma}}
\begin{proposition}
Assume
\[
\sqrt{K_1}+\sqrt{K_2}<I_1.
\]
Then the set
\[
\Gamma=
\{(s_1,d_1,s_2,d_2):
s_1+s_2=I_1,\ 
s_j^2-d_j^2=K_j,\ 
s_j>0\}
\]
is a smooth compact closed orbit of the reduced system
\[
\dot s_1=-\frac12 d_1d_2,\qquad
\dot d_1=-\frac12 s_1d_2,
\qquad
\dot s_2=\frac12 d_1d_2,\qquad
\dot d_2=\frac12 s_2d_1.
\]
Equivalently, if
\[
s_j=\sqrt{K_j}\cosh\theta_j,\qquad
d_j=\sqrt{K_j}\sinh\theta_j,
\]
then \(\Gamma\) is the image of the periodic solution of
\[
\dot\theta_1=-\frac12\sqrt{K_2}\sinh\theta_2,
\qquad
\dot\theta_2=\frac12\sqrt{K_1}\sinh\theta_1
\]
restricted to the compact level set
\[
\sqrt{K_1}\cosh\theta_1+
\sqrt{K_2}\cosh\theta_2=I_1.
\]
\end{proposition}
\begin{proof}
The hyperbolic parametrization maps the level set
\[
\sqrt{K_1}\cosh\theta_1+
\sqrt{K_2}\cosh\theta_2=I_1
\]
onto \(\Gamma\), because it is exactly equivalent to
\[
s_j^2-d_j^2=K_j,\qquad s_j>0,\qquad s_1+s_2=I_1.
\]
Under the strict inequality
\[
\sqrt{K_1}+\sqrt{K_2}<I_1,
\]
this level set is a smooth compact closed curve. The vector field has an equilibrium only at
\[
(\theta_1,\theta_2)=(0,0),
\]
but this point is not on the level set under the strict inequality. Therefore, the vector field is nowhere zero on the level set. Hence, the flow winds around the whole compact component periodically. Its image under the hyperbolic parametrization is precisely \(\Gamma\).
\end{proof}
\section{Relation between \texorpdfstring{$(I_1,K_1,K_2)$ and $(I_3,\mathcal{E})$}{(I1,K1,K2) and (I3,E)} }\label{appendix:relations} 

We derive an explicit correspondence between the invariant description of the curve $\Gamma$ in terms of $(I_1,K_1,K_2)$ and the $(v_+,y)$ Hamiltonian coordinates involving $(I_3,\mathcal{E})$.
\subsection*{Step 1: Expression for $I_3$}

On $\Gamma$ (i.e.\ in the limit $E\equiv 0$), we have
\[
I_3 = m_1 n_1 m_2 n_2.
\]
Using
\[
m_j = \frac{s_j - d_j}{2}, \qquad n_j = \frac{s_j + d_j}{2},
\]
we obtain
\[
m_j n_j = \frac{s_j^2 - d_j^2}{4} = \frac{K_j}{4}.
\]
Therefore,
\begin{equation}
I_3 = \frac{K_1 K_2}{16}.
\end{equation}

\subsection*{Step 2: Expression for $\mathcal{E}$}

Recall that
\[
v_+ = \frac12(n_1+n_2-m_1-m_2) = -\frac12(d_1+d_2),
\]
and
\[
m_1 n_2 = \sqrt{I_3}\,e^y, \qquad
m_2 n_1 = \sqrt{I_3}\,e^{-y}.
\]
Hence,
\[
\sqrt{I_3}\cosh y = \frac12(m_1 n_2 + m_2 n_1).
\]

The energy is
\[
\mathcal{E} = \frac12 v_+^2 + \sqrt{I_3}\cosh y.
\]
Substituting the above expressions, we obtain
\[
\mathcal{E}
=
\frac18(d_1+d_2)^2 + \frac12(m_1 n_2 + m_2 n_1).
\]

Next, observe that
\[
m_1 n_2 + m_2 n_1=\frac12(s_1 s_2 - d_1 d_2).
\]
Therefore,
\[
\mathcal{E}
=
\frac18(d_1+d_2)^2 + \frac14(s_1 s_2 - d_1 d_2).
\]

Expanding and simplifying,
\[
\mathcal{E}
=
\frac18(d_1^2 + d_2^2) + \frac14 s_1 s_2.
\]

Using $d_j^2 = s_j^2 - K_j$, we obtain
\[
\mathcal{E}
=
\frac18(s_1^2 + s_2^2 - K_1 - K_2) + \frac14 s_1 s_2.
\]

Thus,
\[
\mathcal{E}
=
\frac18\bigl((s_1+s_2)^2 - K_1 - K_2\bigr).
\]

Since $s_1+s_2 = I_1$, we conclude that
\begin{equation}
\mathcal{E}
=
\frac{I_1^2 - K_1 - K_2}{8}.
\end{equation}

\subsection*{Final correspondence}

We have established the relations
\begin{equation} \label{eq:I3E}
I_3 = \frac{K_1 K_2}{16},
\qquad
\mathcal{E} = \frac{I_1^2 - K_1 - K_2}{8}.
\end{equation}

Conversely, given $(I_1,I_3,\mathcal{E})$, the quantities $K_1,K_2$ are the two roots of the quadratic equation
\begin{equation}
\lambda^2 - \bigl(I_1^2 - 8\mathcal{E}\bigr)\lambda + 16 I_3 = 0.
\end{equation}

Hence,
\begin{equation}
K_{1,2}
=
\frac12\left[
I_1^2 - 8\mathcal{E}
\pm
\sqrt{
\bigl(I_1^2 - 8\mathcal{E}\bigr)^2 - 64 I_3
}
\right].
\end{equation}

This establishes a one-to-one correspondence between the invariant descriptions $(I_1,K_1,K_2)$ and $(I_1,I_3,\mathcal{E})$, up to permutation of $K_1$ and $K_2$.
\section{Periodic Dynamics on \texorpdfstring{$\Gamma$}{Gamma} and Drift Velocities}
\label{appendix:positions}

Let $G>0$ be fixed and introduce the notation
\[
m:=-k^2,\qquad n(G):=-\frac{4k^2G}{(1+G)^2}.  
\]
For the asymptotic orbit, we take
\[
G=G^\infty.
\]

Define
\[
\mathcal P(u):=\int_0^u\frac{d\xi}{1-n\sn^2(\xi\mid m)}.
\]
Our initial goal is to compute the primitives $S_j^\Gamma(t)=\int_0^t s_j^\Gamma(\tau)\, d\tau$ associated with the orbit $\Gamma$.  
Since $s_1^\Gamma+s_2^\Gamma=I_1$ we immediately obtain the relation 
\begin{equation} \label{eq:S1S2}
S_1^\Gamma(t)=I_1t-S_2^\Gamma(t). 
\end{equation} 
Using \autoref{prop:m1n1m2n2-d}, we obtain
\begin{align}
S_2^\Gamma(t)=\int_0^t s_2^{\Gamma}(\tau)d\tau=\frac {I_1}{2}\int_0^t\left(\frac{1}{1+G e^{-y(\tau)}}+\frac{1}{1+G e^{y(\tau)}}\right)d\tau+\int_0^t \left(\frac{\dot y(\tau)}{1+Ge^{-y(\tau)}}-\frac{\dot y(\tau)}{1+G e^{y(\tau)}}\right)d\tau.  
\end{align}
We will work out each integral individually, and combine them at the end. The simpler integral is the one involving $\dot y$. We apply the substitution $u=e^{y}$ to get 
\[
 \int_0^t \left(\frac{\dot y}{1+G e^{-y}}-\frac{\dot y}{1+G e^y}\right)\, d\tau =\ln\frac{1+2G\cosh y(t)+G^2}{1+2G\cosh y(0)+G^2}. 
\]

The integral 
$
\frac {I_1}{2}\int_0^t\left (\frac{1}{1+G e^{-y}}+\frac{1}{1+Ge^y}\right)\, d\tau 
$
is more involved. 
We first combine the denominators to obtain an integral involving $\cosh{y}$:
\begin{align*}
\frac {I_1}{2}\int_0^t\left(\frac{1}{1+Ge^{-y(\tau)}}+\frac{1}{1+Ge^{y(\tau)}}\right)d\tau
&=\frac{I_1}{2}\int_0^t \frac{2+2G \cosh{y(\tau)}}{1+2G\cosh{y(\tau)}+G^2}d\tau. 
\end{align*}
Then, after applying the energy equation \eqref{eq:Enorm}, 
 we reduce it to an $\mathrm{sn}^2$ integral of the following form:
\begin{align*}
    \frac{I_1}{2}\int_0^t \frac{2+2G \cosh{y(\tau)}}{1+2G \cosh{y(\tau)}+G^2}d\tau=\frac{I_1}{2}t + \frac{I_1}{2} \frac{1-G}{1+G}\int_0^t\frac{1}{1+\frac{4k^2G}{(1+G)^2}\mathrm{sn}^2(\omega(\tau-t_0)|m)}d\tau
\end{align*}
After the substitution \(u=\omega(\tau-t_0)\), and noting that $n(G)=-\frac{4k^2G}{(1+G)^2}$, the remaining integral can be expressed in terms of $\mathcal{P}(u)$.  More precisely, we obtain:

$$
\frac {I_1}{2}\int_0^t\left (\frac{1}{1+G e^{-y}}+\frac{1}{1+Ge^y}\right)\, d\tau =\frac{I_1}{2}t+ \frac{I_1}{2\omega} \frac{1-G}{1+G}\left (\mathcal{P}(\omega(t-t_0))-\mathcal{P}(-\omega t_0)\right).  
$$

 This gives us an intermediate result
 \begin{align*} 
 S_1^\Gamma(t)&=\frac{I_1}{2}t -4\kappa \left[\mathcal{P}(\omega (t-t_0))-\mathcal{P}(-\omega t_0)\right] -\ln\frac{1+2G\cosh y(t)+G^2}{1+2G\cosh y(0)+G^2} \\
 S_2^\Gamma(t)&=\frac{I_1}{2}t + 4\kappa \left[\mathcal{P}(\omega (t-t_0))-\mathcal{P}(-\omega t_0)\right] +\ln\frac{1+2G\cosh y(t)+G^2}{1+2G\cosh y(0)+G^2}, \\
 \text{ where }  \kappa&:= \frac{1}{8\omega}A(G), \qquad \textrm{and } A(G):=I_1\frac{1-G}{1+G}.  
 \end{align*}

Now, for $k>0$ we use the fact that the integrand of $\mathcal{P}(u)$, namely, 
$$
f(u)=\frac{1}{1-n\mathrm{sn}^2(u\mid m)}, 
$$ 
is periodic with period $2K(m)$. Thus, we can write $\mathcal{P}(u)$ as a linear term plus a periodic correction $P_{\mathrm{ell}} (u)$ defined by:
\[
    \mathcal{P}(u)=\alpha u+P_{\mathrm{ell}}(u).  
    \]
The constant $\alpha$ is the average value of the integrand over the period and is given by:  
\[
    \alpha=\frac{1}{2K(m)}\int_0^{2K(m)} \frac{1}{1-n\mathrm{sn}^2(u\mid m)}\, du=\frac{\Pi(n(G)\mid m)}{K(m)}, 
    \]
where $\Pi$ is the complete elliptic integral of the third kind \cite{AS64}[17.7.2].  
Since $u=\omega(t-t_0)$, the period in the time variable is $T_A=\frac{2K(m)}{\omega}$, yielding:
\begin{subequations} \label{eq:S1S2bis}
\begin{align} 
    S_1^\Gamma(t) &=[\frac{I_1}{2}-4\kappa\omega\frac{\Pi(n(G)  \mid m)}{K(m)}]t-P_{\mathrm{per}}(t;G), \\
    S_2^\Gamma(t) &=[\frac{I_1}{2}+4\kappa\omega\frac{\Pi(n(G) \mid m)}{K(m)}]t+P_{\mathrm{per} }(t;G), 
\end{align}
\end{subequations}
with 
\[
P_{\text{per}}(t;G):=4\kappa \left[P_{\text{ell}}(\omega(t-t_0))- P_{\text{ell}}(-\omega t_0)\right]+\ln \frac{1+2G\cosh y(t)+G^2}{1+2G\cosh y(0)+G^2}, 
\]
which is $T_A$-periodic and satisfies $P_{\mathrm{per}}(0; G)=0$. 

The time averages over one period, $\langle s_j^\Gamma \rangle:=\frac{1}{T_A}\int_0^{T _A}s_j^\Gamma(\tau)\, d\tau$, are now easily computed to be 
\begin{subequations} 
\begin{align}
\langle s_1^\Gamma\rangle=\frac{I_1}{2}-4\kappa\omega
\frac{\Pi(n(G)|m)}
     {K(m)}, \label{eq:s1ave}\\
\langle s_2^\Gamma\rangle
=
\frac{I_1}{2}
+
4\kappa\omega
\frac{\Pi(n(G)|m)}
     {K(m)}, \label{eq:s2ave}
\end{align}
\end{subequations} 

and,
\begin{equation} \label{eq:s2s1ave} 
\langle s_2^\Gamma-s_1^\Gamma\rangle
=A(G)
\frac{\Pi(n(G)|m)}
     {K(m)}.
\end{equation} 

Since
$n\leq 0$, the complete elliptic integral $\Pi(n(G)|m)$ is positive. Consequently, the sign of the averaged quantity
$
\langle s_2^\Gamma-s_1^\Gamma\rangle
$
is determined entirely by the sign of $A(G)$, hence by the sign of $1-G.$

We defined the asymptotic drift velocities $ v_j^\infty:= \frac14 \langle s_j^\Gamma \rangle$.  Thus 
\begin{align} 
v_1^\infty&=\frac{I_1}{8}-\kappa\omega\frac{\Pi(n(G)|m)}{K(m)},  \\
v_2^\infty&=\frac{I_1}{8}+\kappa\omega\frac{\Pi(n(G)|m)}{K(m)}. 
\end{align}  
As a consistency check, letting $k\rightarrow0$ gives
\[
n\rightarrow0, \qquad \frac{\Pi(0|0)}{K(0)}=1,
\]
so that
\[
\langle s_1^\Gamma\rangle=\frac{I_1G}{1+G},\qquad \langle s_2^\Gamma\rangle=\frac{I_1}{1+G}.
\]

Using
\[
G=\sqrt{\frac{K_1}{K_2}}, \qquad  I_1=\sqrt{K_1}+\sqrt{K_2},
\]
one recovers
\[
\langle s_1^\Gamma\rangle=\sqrt{K_1}, \qquad \langle s_2^\Gamma\rangle=\sqrt{K_2},
\]
which agrees with the degenerate (Camassa--Holm) limit.
We conclude this appendix with the following elementary remark, which appears in several places in the main body of the paper.  
\begin{remark} \label{rem:AppC/G} 
Throughout this appendix, $G$ was treated as a positive constant.
Although the computation was motivated by
\[
G^\infty=\sqrt{\frac{K_1^\infty}{K_2^\infty}},
\]
the algebraic formulas remain valid for every fixed $G>0$. For arbitrary $G$ though, the resulting periodic functions serve as frozen-parameter comparison functions; they represent an actual orbit of the $E=0$ dynamical system only
when $G$ satisfies the compatibility condition $I_1^2-8\mathcal{E}=4\sqrt{I_3}(G+G^{-1})$ (see \autoref{lem:v+}). \end{remark}

\def\cydot{\leavevmode\raise.4ex\hbox{.}}
  \def\cydot{\leavevmode\raise.4ex\hbox{.}}

\bibliographystyle{plain} 
%\bibliography{KS}

\end{document}